\documentclass{article}

\usepackage{graphicx,verbatim,setspace,lineno}
\usepackage{willcommands}
\usepackage{amsmath,amssymb,amsthm,bm}
\usepackage{pdfpages}
\usepackage{latexsym,color,mdframed,textcomp}
\usepackage[round]{natbib}
\usepackage[lined,ruled]{algorithm2e}

\begin{document}
%\doublespacing
%\linenumbers

\newtheorem{theorem}{Theorem}
\newtheorem{corollary}[theorem]{Corollary}
\newtheorem{lemma}[theorem]{Lemma}
\newtheorem{observation}[theorem]{Observation}
\newtheorem{proposition}[theorem]{Proposition}
\newtheorem{definition}[theorem]{Definition}
\newtheorem{claim}[theorem]{Claim}
\newtheorem{fact}[theorem]{Fact}
\newtheorem{assumption}[theorem]{Assumption}

\title{Bias Correction in Species Distribution Models:
  Pooling Survey and Collection Data for Multiple Species}
\author{William Fithian, Jane Elith, Trevor Hastie, and David A. Keith}
\maketitle

\newcommand{\D}{\mathcal{D}}
\renewcommand{\grad}{\nabla}
\renewcommand{\S}{\mathcal{S}}
\newcommand{\T}{\mathcal{T}}
\newcommand{\bbX}{\mathbb{X}}
\newcommand{\cX}{\mathcal{X}}
\newcommand{\cZ}{\mathcal{Z}}
\newcommand{\red}{\textcolor{red}}
\newcommand{\PA}{\text{PA}}
\newcommand{\PO}{\text{PO}}
\newcommand{\BG}{\text{BG}}

\begin{abstract}
  \begin{enumerate}
  \item Presence-only records may provide data on the distributions of
    rare species, but commonly
    suffer from large, unknown biases due to their
    typically haphazard collection schemes.  Presence-absence or count data
    collected in systematic, planned surveys are more reliable but
    typically less
    abundant.
  \item We proposed a probabilistic model to allow for joint
  analysis of presence-only and survey data to exploit their
  complementary strengths.
  Our method pools presence-only and presence-absence data for many
  species and maximizes a joint likelihood, simultaneously estimating
  and adjusting for the sampling bias affecting the presence-only data.
  By assuming that the sampling bias is the same for all species,
  we can borrow strength across species to
  efficiently estimate the bias and improve our inference from
  presence-only data.

  \item We evaluate our model's performance on data for
  36 eucalypt species in southeastern Australia.  We find that
  presence-only records exhibit a strong
  sampling bias toward the coast and toward Sydney, the largest
  city. Our data-pooling technique substantially improves
  the out-of-sample predictive performance of our model when the
  amount of available presence-absence data for a given species is
  scarce.

\item If we have only presence-only data and
  no presence-absence data for a given species, but both types of data
  for several other species that suffer from the same spatial
  sampling bias, then our method can obtain an unbiased estimate of
  the first species' geographic range.
  \end{enumerate}
\end{abstract}

\section{Introduction}

Presence-only data sets \citep{pearce2006modelling} are  key sources
of information about factors that influence the habitat relationships
and distributions of plants and animals, and
% of information about
% the determinants of habitat suitability, and
analyzing them accurately is
crucial for successful wildlife management policy.
Examples include specimen collection data from museums and herbaria,
and atlas records maintained by government agencies and non-government
organizations. Often, these are the most abundant and freely
available data on species occurrence.
However, sampling
bias often confounds efforts to reconstruct species distributions.

Recent work has shown that several of the most popular methods for
species distribution modeling
with presence-only data are equivalent or nearly equivalent to
each other, and may be motivated by an underlying inhomogeneous
Poisson process (IPP) model \citep{AartsIPP, WartonIPP,
  renner2013equivalence, fithian2013finite}.  In effect, all of these
methods estimate the distribution of species {\em sightings}
(i.e. of presence-only records) under an exponential family model
for the species distribution \citep{fithian2013finite}.
Because presence-only data are commonly
collected opportunistically, the sightings distribution is typically
biased toward regions more frequented by whoever is collecting the
data.  Thus, it may be a poor proxy for the distribution of {\em all}
organisms of that species, sighted or unsighted.

Presence-absence and other data sets collected via systematic surveys
do not typically suffer from such bias. Even if (say) survey sites
cluster near a major city, the data will
contain more presences {\em and} more absences there.
% in the oversampled
% region, which typically will not lead to overestimation of species
% abundance in the oversampled region.
% By contrast, presence-only data carry no accompanying records of
% unsuccessful searches, so that uneven
% sampling effort can be mistaken for differences in abundance.
Unfortunately, if the species under study is rare, presence-absence
data may carry little information about its species
distribution.  In this article we consider a large presence-absence
data set on eucalypts in southeastern Australia.  Although there are
over 32,000 sites, 4 of the 36 species we consider are present in
fewer than 20 of the survey sites. Presence-only
data for rare species, suitably adjusted for bias, can supplement
survey data.

We propose a natural extension of the IPP model for
single-species presence-only data, with a view toward estimating and
adjusting for sampling bias.  In particular, our method brings
other sources of data --- presence-only and
presence-absence data for multiple species --- to bear on the problem,
by incorporating them into a single joint probabilistic model to
estimate and adjust for the bias.
Some of the most popular approaches to analysis of
presence-absence or presence-only data for one species are
special cases of our joint approach. We evaluate our model using both
presence only and presence-absence data for a set of eucalypt species from
southeastern Australia.  An R package implementing
our method, \texttt{multispeciesPP}, is available in the public github
repository \texttt{wfithian/multispeciesPP}.

\subsection{The Inhomogeneous Poisson Process Model}

The starting point for our model is the random set $\S$ of point
locations of {\em all} individuals of a given species in some
geographic domain $\D$. In spatial statistics, such a random set is
called a {\em point process}, and we will call the set $\S$
the {\em species process}. Typically $\D$ is a bounded two-dimensional
region.

The IPP model is a probabilistic model for the random set
$\S=\{s_i\}\sub \D$. It is characterized by an {\em intensity function}
$\lambda(s)$, which maps sites in $\D$ to non-negative real
numbers. Informally, $\lambda(s)$ quantifies how many $s_i$
are likely to occur near $s$.

For any sub-region $A$ within $\D$, let $N_\S(A)$ denote the number of
points $s_i\in \S$ falling into $A$. If $\S$ is an IPP with
intensity $\lambda$, then $N_\S(A)$ is a Poisson random variable with
mean
\begin{linenomath}\begin{equation}\label{eq:defLambda}
  \Lambda(A) = \int_A \lambda(s)\,ds.
\end{equation}\end{linenomath}
For non-overlapping sub-regions $A$ and $B$, $N_\S(A)$ and $N_\S(B)$ are
independent.

If $A$ is a quadrat centered at $s$, small enough that $\lambda$
is nearly constant over $A$, then $\Lambda(A) \approx \lambda(s)|A|$, where
$|A|$ represents the area of sub-region $A$.
Therefore, the intensity $\lambda(s)$ represents the expected species
count per unit area near $s$. The integral $\Lambda(\D)$ over the
entire study region is the expectation of $N_\S(\D)$, the population
size.

We can normalize $\lambda(s)$ to obtain the function ${p_\lambda(s) =
  \frac{1}{\Lambda(\D)} \lambda(s)}$, which integrates to one and
represents the probability distribution of individuals.
An IPP may be defined equivalently as an independent random sample
from $p_\lambda(s)$ whose size $N_\S(\D)$ is itself a Poisson
random variable with mean $\Lambda(\D)$.
Conditional on the number $N_\S(\D)$ of points, their locations
$s_1,\ldots,s_{N_\S(\D)}$ are independent and identically distributed (i.i.d.)
draws from $p_\lambda(s)$. We
call the intensity $\lambda(s)$ of $\S$ the
{\em species intensity} and the density function $p_\lambda(s)$
the {\em species distribution}.
See \cite{cressie1993} for a more in-depth discussion of Poisson
processes and other point process models.

The first panel of Figure~\ref{fig:toyIPP} shows a realization
of a simulated IPP on a rectangular domain.  The
background coloring shows the intensity, and the black circles denote
the $s_i\in \S$.  Relatively more of the black circles occur in the
green region where the intensity is highest.

\begin{figure}
  \centering
  \includegraphics[width=.93\textwidth]{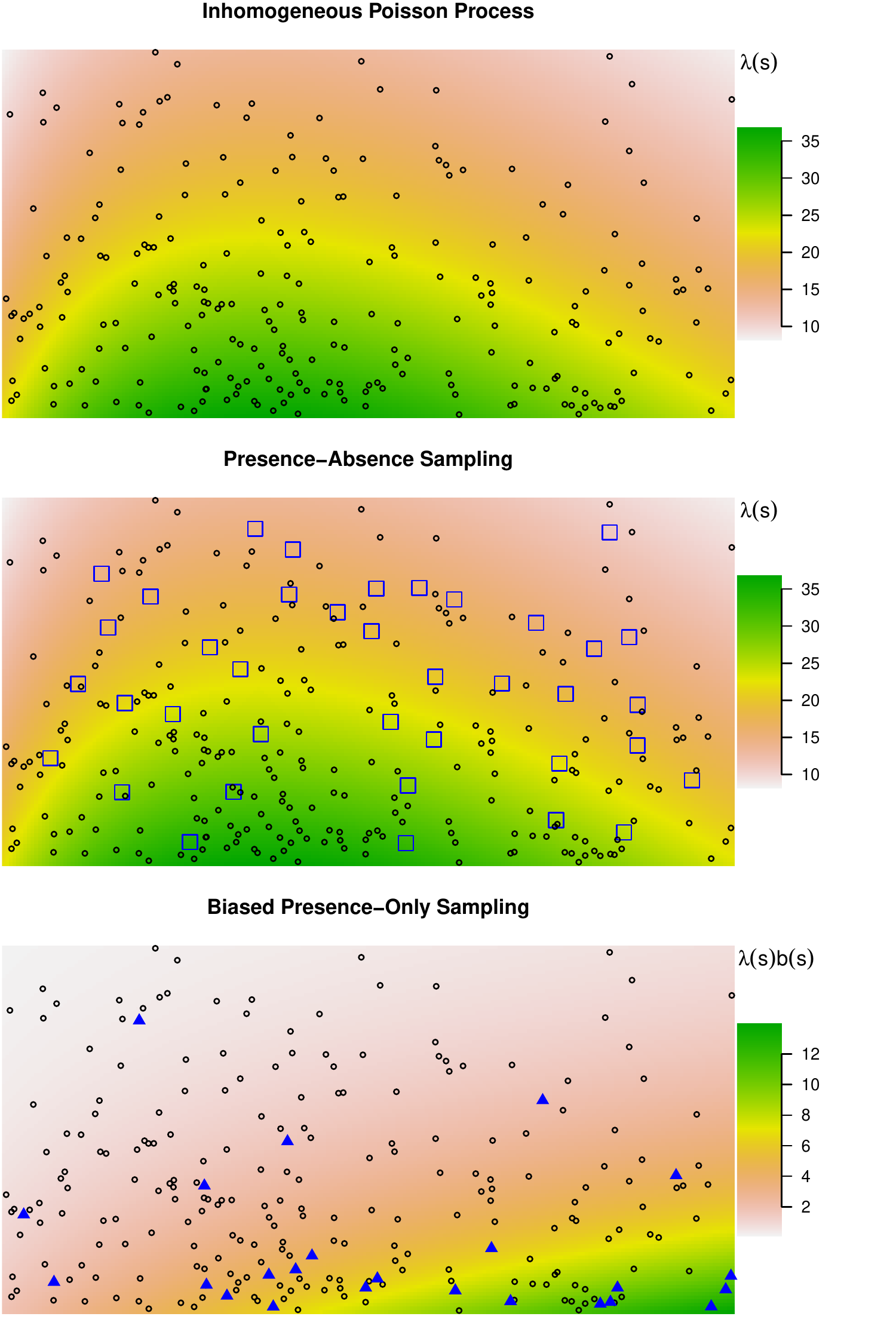}
  \caption{A Poisson process with two different sampling schemes
    representing our models for presence-absence and presence-only
    data.  The top panel represents the species process against a heat
    map of the species intensity $\lambda(s)$.  The second panel
    depicts presence-absence or other systematic survey methods:
    quadrats (blue squares) are surveyed and organisms counted in each
    one.  The third panel depicts biased presence-only
    sampling, with the blue triangles indicating the presence-only
    process, a small and unrepresentative subset of the species
    process.  The heat map shows the presence-only intensity ${\lambda(s)
      \cdot b(s)}$.}
  \label{fig:toyIPP}
\end{figure}

In modern ecological data sets each site in the domain
has associated environmental covariates $x(s)$ measured in the field,
by satellite, or on biophysical maps.
These are assumed to drive the intensity $\lambda(s)$.
It is convenient to model the intensity using a log-linear
form for its dependence on the features:
\begin{linenomath}\begin{equation}\label{eq:loglin}
  \log \lambda(s) = \alpha + \beta' x(s)
\end{equation}\end{linenomath}
The linear assumption in \eqref{eq:loglin}
is not nearly as restrictive as it might at first seem. The
feature vector $x(s)$ could contain basis expansions
such as interactions or spline terms allowing us to fit highly
nonlinear functions of the raw features
(see, e.g., \cite{ESL}).

Unfortunately, we cannot observe the entire species process
$\S$, but we can glimpse it incompletely in various ways. The most
straightforward and reliable way to learn about $\S$ is with presence-absence or count
sampling via systematic surveys, as depicted in the second panel of
Figure~\ref{fig:toyIPP}. In survey data, an ecologist visits numerous
quadrats $A_i$ throughout $\D$ (the blue
squares) and records the species' occurrence or count $N_\S(A_i)$ at each one.

Presence-only data is a less reliable but often more abundant
source of information about $\S$. We discuss our model for
presence-only data in the next section.

\subsection{Thinned Poisson Processes}

The {\em presence-only process} $\T$
comprises the set of all individuals observed by opportunistic
presence-only sampling. Assuming they are identified
correctly (not always a given), $\T$ is the subset of $\S$
that remains after the unobserved individuals are removed --- or
{\em thinned}, in statistical language.

We propose a simple model for how $\T$ arises given $\S$:
an individual at location $s_i\in\S$ is
included in $\T$ (is observed) with probability $b(s_i)\in [0,1]$,
independently of all other individuals.  The function $b(s)$, which we
call the {\em sampling bias}, represents the expected fraction
(typically small) of all organisms near location $s$ that are counted
in the presence-only data. As a result of the biased thinning,
individuals in areas with relatively large $b(s)$ will tend to be
over-represented relative to areas with small $b(s)$.

It can be shown that marginally
\begin{linenomath}\begin{equation}\label{eq:thinPP}
  \T\sim \text{IPP}(\lambda(s) \; b(s))
\end{equation}\end{linenomath}
For a formal proof, see \citet{cressie1993} section 8.5.6, p. 689.
Informally, a small sub-region $A$ centered at $s$ contains on
average $|A|\lambda(s)$ individuals, of which on average
$|A|\lambda(s)\;b(s)$ are observed. If two sites $s_1$ and $s_2$ have
the same intensity $\lambda(s_1)=\lambda(s_2)$, but $b(s_1) =
2b(s_2)$, then~(\ref{eq:thinPP}) means the presence-only data will
have about twice as many records near $s_1$ as $s_2$.

The third panel of Figure~\ref{fig:toyIPP} displays a thinning
of the Poisson process shown in the first two panels. The thinned
process $\T$, consisting of the solid blue triangles, is shown against a
heat map of the biased intensity $\lambda(s) \; b(s)$.

Sampling bias in presence-only data is not a subtle phenomenon. By our
estimates in Section~\ref{sec:euca}, $b(s)$ ranges from about $3\times
10^{-3}$ near Sydney to about $3\times 10^{-7}$ in the more rugged inland areas
of southeastern Australia --- a dynamic range of 10,000.

Some of the most popular methods for analyzing presence-only data
are based explicitly or implicitly on fitting a log-linear IPP model
for the process $\T$. It is clear from~(\ref{eq:thinPP}) that this
approach effectively yields an estimate of the presence-only
  intensity $\lambda(s)\; b(s)$ and not the species intensity
$\lambda(s)$. These estimates may be dramatically inaccurate if treated
as estimates of the species intensity or species distribution.

In the case of presence-only data, $b(s)$ typically depends
on the behavior of whoever is collecting the presence-only data.
When sampling bias is thought to depend mainly on a few measured
covariates $z(s)$ (such as distance from a road network
or a large city), several authors have proposed modeling presence-only
data directly as a thinned Poisson process \citep{chakraborty2011point, fithian2013finite,
  warton2013model, hefley2013nondetection}.  A similar method was
proposed in \citet{dudik2005correcting} in the context of the Maxent
method, and \citet{zaniewski2002predicting} similarly propose weighting background points in presence-background GAMs according to a model for their likelihood of appearing as absences in presence-absence data.

If both $\lambda$ and $b$ are modeled as
log-linear in their respective covariates,
then we have
\begin{linenomath}\begin{equation}\label{eq:loglinbias}
  \log \left(\lambda(s) \; b(s)\right) = \alpha + \beta'x(s) +
  \gamma + \delta'z(s)
\end{equation}\end{linenomath}
Modeling the bias as above amounts to estimating the effects of the
variables $x(s)$ in a generalized linear model (GLM) for the Poisson
process $\T$, while adjusting for control variables $z(s)$. We will
refer to it as the ``regression adjustment'' strategy.\footnote{Because
  $b(s)$ is a probability, readers
  familiar with logistic regression may wonder why we model
  $b(s)=e^{\gamma+\delta'z(s)}$ instead of $b(s) =
  \frac{e^{\gamma+\delta'z(s)}}{1+e^{\gamma+\delta'z(s)}}$. When
  $b(s)$ is close to zero, the denominator
  $1+e^{\gamma+\delta'z(s)}\approx 1$ and the two models roughly
  coincide. We use the log-linear form
  because it leads to the convenient log-linear form for the presence-only
  intensity in~(\ref{eq:loglinbias}).}

\subsection{Identifiability, Abundance, and the Role of $\gamma$}\label{sec:ident}

Modeling presence-only data as a thinned Poisson process as
in~\eqref{eq:loglinbias} sheds light on why it is so difficult to
obtain useful estimates of presence probabilities: at best,
presence-only data reflect relative intensities and not properly
calibrated probabilities of occurrence.
If the covariates comprising
$x$ and $z$ are distinct and have no perfect linear dependencies on one
other, then $\beta$, $\delta$, and the sum $\alpha+\gamma$ are identifiable,
but individually $\alpha$ and $\gamma$ are not.

To see why, consider
\begin{enumerate}
\item a presence-only process governed by species
  process parameters $(\alpha, \beta)$ and thinning parameters
  $(\gamma,\delta)$, and
\item an alternative process with $\alpha$ replaced by
  ${\tilde\alpha = \alpha + \log 2}$ (trees are twice as abundant
  overall) and $\gamma$ replaced by ${\tilde \gamma
    = \gamma - \log 2}$ (the chance of observing any given tree is halved
  overall).
\end{enumerate}
\eqref{eq:loglinbias} means that the probability distribution of the
thinned process $\T$ is {\em identical} in these two cases.
Therefore, no matter how much data we
collect, we can never distinguish parameters
$(\alpha,\beta,\gamma,\delta)$ from
$\left(\tilde\alpha, \beta, \tilde\gamma,\delta\right)$ on the basis of
presence-only data alone.

Because $\beta$ is identifiable, we can use presence-only data
alone to obtain an estimate for
$\lambda(s)$ up to the unknown proportionality constant $e^{\alpha}$;
in other words, we can estimate the species distribution
$p_\lambda$ but not the species intensity $\lambda$.
If the model is correctly
specified, then likelihood estimation gives an asymptotically
unbiased estimate of the model's parameters \citep[see
e.g.][]{lehmann1998theory}.

The species intensity $\lambda(s)$ is the product of the
species distribution $p_\lambda(s)$ and the overall abundance
$\Lambda(\D)$.  Predicting the probability that a species is present in some
new quadrat $A$ requires information about both.
Considerable attention has focused on whether or not we can obtain
plausible estimates of abundance or of presence probabilities based on
presence-only data alone.  Methods like Maxent and presence-background logistic
regression explicitly estimate $p_\lambda(s)$, but require an
externally-given specification of the overall abundance if presence
probabilities are required \citep[for example, Maxent's ``logistic
output,'' see][]{elith2011statistical}.  Other methods attempt to
estimate presence probabilities \citep{lele2006weighted,
  royle2012likelihood}, but estimates can be highly variable and
non-robust to minor misspecifications of the modeling assumptions
\citep{ward2009em, hastie2013inference}.

One of the purported advantages of the IPP as a model for presence-only
data is that it does yield an estimate of overall abundance because
its intercept term is identifiable \citep{renner2013equivalence}. However,
\citet{fithian2013finite}
show that the maximum likelihood estimate of $\widehat{\Lambda}(\D)$
obtained from that model is exactly the
number of presence-only records in the data set, so it should not be
regarded as an estimate of the overall abundance.

\subsection{Challenges for Regression Adjustment Using Presence-Only
  Data}\label{sec:regadj}

Regression adjustment works best when the control variables $z(s)$ are not too
correlated with $x(s)$, the covariates of interest.  If e.g. $x_1(s)$ and $z_2(s)$
are highly correlated, then we can increase $\beta_1$ and decrease
$\delta_2$ without altering the model's predictions much.  As a result, we may
need a great deal of data to distinguish the effects of $\beta_1$ and
$\delta_2$, and hence to tease apart $\lambda$ and $b$.

Unfortunately, correlation between $x$ and $z$ is all too common, in
part because humans respond to many of the same
covariates as other species do.  For example, in southeastern
Australia, major
population centers lie along the eastern coastline, but many important
climatic variables are also correlated with distance from the coast.
Figure~\ref{fig:bc02map} plots  the mean diurnal temperature range
over a region of southeastern Australia, juxtaposed against our
fitted bias from the model we will fit in Section~\ref{sec:euca} using
pooled data.  The bias is almost perfectly confounded with temperature
range, making estimation highly variable even if the model is
correctly specified.

Another difficulty of regression adjustment in real-world settings is that
our functional form is always misspecified.  In particular, it may be
difficult to obtain good features in modeling the bias. Suppose for example
that $x_1(s)$ is highly correlated with $z_2(s)^2$, which (unbeknown to us)
is an important bias covariate.  If we fit our model without including
$z_2(s)^2$, then the $\beta_1 x_1(s)$ term may serve as a proxy for
the missing quadratic effect, biasing our estimate $\hat\beta_1$.

In practice we expect there to be missing variables as well as
unaccounted-for nonlinearities and interactions in our models for both
the species intensities and the bias alike.
We can mitigate this sort of problem by adding more basis functions to
$z(s)$, but as the dimension of the model increases, the
standard errors of our estimates will tend to increase along with it.

If any bias covariates coincide with $x$ variables ---
e.g., if rugged terrain is undersampled due to inaccessibility {\em
  and} has an effect on a species'
abundance --- then the corresponding coordinates of $\beta$ and
$\delta$ are unidentifiable no
matter how much presence-only data we collect.

For all its difficulties, regression
adjustment on presence-only data is often preferable to no
adjustment, and may be the best option when unbiased survey data is
unavailable. Still, when some components of $x$ are nearly or
completely confounded by $z$, a small quantity of unbiased
data can go a long way, because it may provide the only solid
information to distinguish true effects from bias
effects (see, e.g., Figure~\ref{fig:confEll}).
This motivates a method that can combine both biased and unbiased data to
exploit the strengths of each.

\begin{figure}
  \centering
  \includegraphics[width=\textwidth]{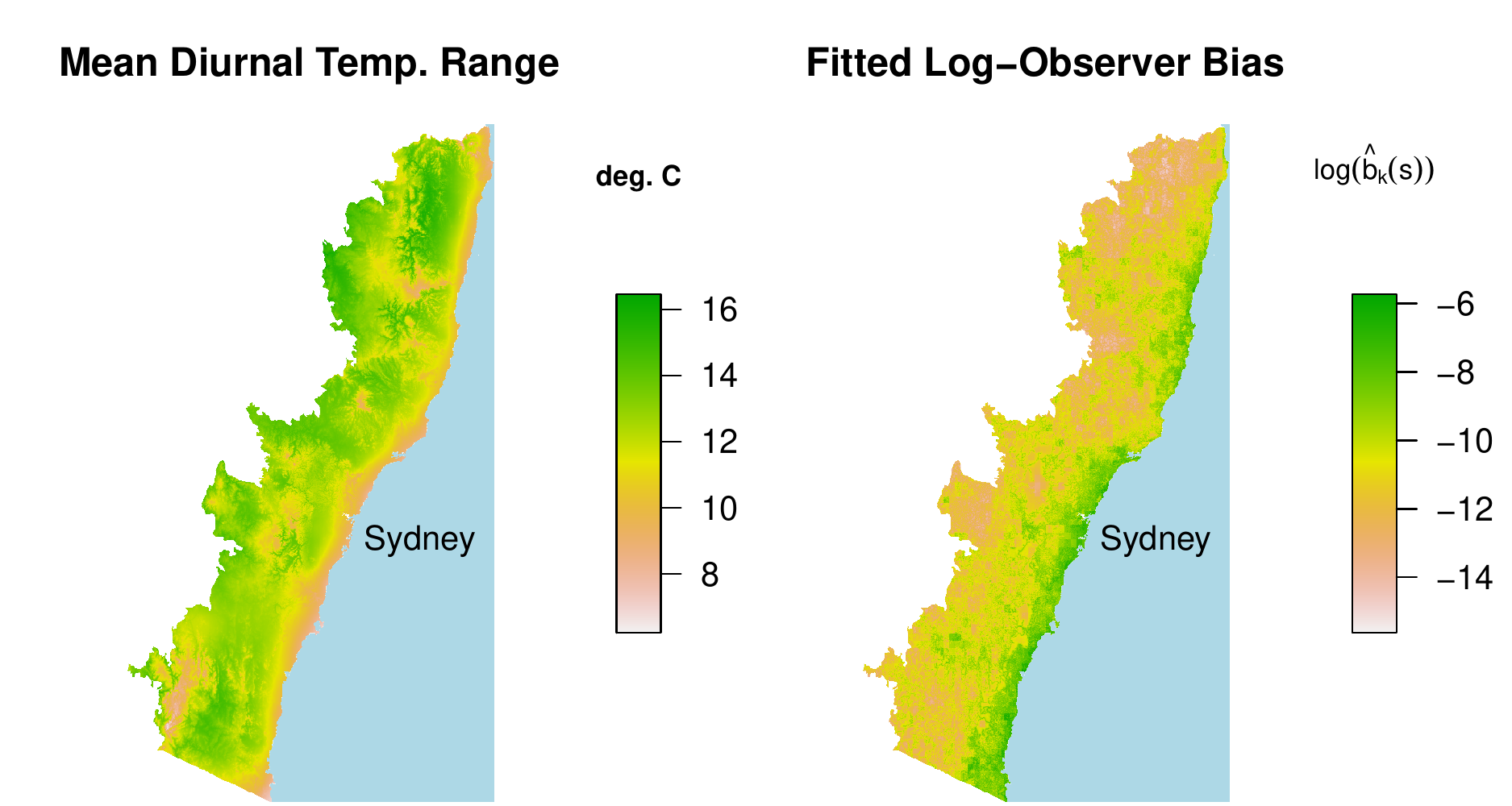}
  \caption{Mean diurnal temperature range in a coastal region of
    southeastern Australia, juxtaposed against our model's fitted sampling
    bias. The black triangles show the locations of cities.
    Because most people live near the coast,
    sampling bias is highly correlated with distance from
    the coastline.  Unfortunately, so are many important climatic
    variables.  Because these variables are almost perfectly
    confounded with bias, it is very difficult to correct for sampling bias
    using presence-only data alone.}
  \label{fig:bc02map}
\end{figure}

\section{A Unifying Model for Presence-Absence and Presence-Only Data}\label{sec:model}

The above discussion motivates a natural unifying model to explain
both presence-absence and presence-only data for many species at once,
which we discuss in detail here.

Assume we are equipped with a real-valued
environmental covariate function $x(s)$, which takes
values in $\R^p$, and bias covariate function $z(s)$, which takes
values in $\R^r$. $x(s)$ and $z(s)$ represent features thought
respectively to influence habitat suitably and heterogeneity in
sampling effort.
In general some variables may appear in both $x$ and $z$.

Let $m$ denote the total number of species for which we have data.
Let $\S_k$ and $\T_k$ denote the species and presence-only processes
for species $k = 1,\ldots,m$.
Our data set consists of two distinct types of observations for each
species, presence-absence or count survey sites and presence-only sites.
By modeling each of the two sampling schemes in terms
of the latent species processes, we can use likelihood methods to pool
data from each.
We adopt the convention of indexing observations by the letter
$i$, variables by the letter $j$, and species by the letter $k$.

Each observation $i$ is associated with a site $s_i\in \D$,
as well as covariates $x_i = x(s_i)$ and $z_i = z(s_i)$.  For
survey sites, $s_i$ represents the centroid of a quadrat
$A_i$. At survey site $i$ we observe counts $N_{ik}=N_{\S_k}(A_i)$ or
binary presence/absence indicators $y_{ik}$, with $y_{ik}=1$ if
$N_{ik}>0$ and $y_{ik}=0$ otherwise.
% Typically, data for every
% species is collected in every quadrat, but it is no problem if
% some of the $N_{ik}$ are missing.

\subsection{Joint Log-Linear IPP Model for Multispecies Data}

For species $k$, we propose to model $\S_k\sim
\text{IPP}(\lambda_k(s))$, with $\T_k\sim \text{IPP}(\lambda_k(s)\;
b_k(s))$ obtained by thinning $\S_k$ via $b_k(s)$.
Both $\S_k$ and $\T_k$ are assumed to be independent across species
with log-linear intensity $\lambda_k$ and bias $b_k$:
\begin{linenomath}\begin{align}\label{eq:logLinAssumption}
  \log \lambda_k(s) &= \alpha_k + \beta_k' x(s)\\\label{eq:propBiasAssumption}
  \log b_k(s) &= \gamma_k + \delta' z(s).
\end{align}\end{linenomath}
Note that $\delta$ is the only model
parameter {\em not} allowed to vary across species --- in
other words, the functions $b_1(s),\ldots,b_m(s)$ are all
assumed to be proportional to one another. We call this the {\em
  proportional bias assumption,} and it lets us
pool information across all $m$ species to
jointly estimate the selection bias affecting the presence-only
data. When $m$ is large, this affords us the option of working with a
more expansive model for the bias term, reducing the resulting bias
in our estimates for the $\alpha_k$ and $\beta_k$, which are typically
of greater scientific interest.

Scientifically, the proportional-bias assumption corresponds to
a belief that the biasing process has more to do with the behavior of
observers than of plants and animals.  Put simply, if one species is
oversampled near Sydney by a factor of five relative to another
region with similar features,
the most likely explanation is that observers spend one fifth as
much time in the second region as they do in Sydney.
In that case, we should expect other species to be undersampled in the
second region by roughly the same factor relative to Sydney.

The proportional-bias assumption could well be violated if, for
example, most of the
observers collecting samples for species 1 reside in Sydney and those
collecting samples for species 2 reside in Newcastle.  Even
under the best of circumstances, this modeling assumption (like the
other assumptions we have made) is an
idealization of the truth, but it can be a very useful one if it is
not too badly wrong.  In
Section~\ref{sec:euca} we provide evidence that the proportional bias
model improves out-of-sample reconstruction of the species intensity.

We allow $\gamma_k$, the proportionality constant of the sampling
bias, to vary by species, representing a species-dependent effect on
overall sampling effort.
This allows us to account for observers systematically oversampling
some species relative to others.  For example, if an ecologist is
collecting samples in a forest, she may preferentially collect samples
from rarer species.  In Section~\ref{sec:euca} we give some evidence
that sampling effort does indeed vary significantly by species in just
this way.  The cost of letting $\gamma_k$ vary by species is that $\alpha_k$ is
unidentifiable unless we have some presence-absence data for species
$k$.  Consequently, we can estimate the species distribution
$p_\lambda(s)$, but not the overall abundance $\Lambda(\D)$, unless we
have some presence-absence or count data for species $k$.

While our paper was in press we learned of concurrent and independent work by \citet{giraud2014capitalising} which uses a similar Poisson thinning model to combine survey and collection data on discrete domains.

\subsection{Induced Model for Survey Data}

Survey data provides information about the species process $\S_k$
restricted to the survey quadrats. If the point locations of each
individual within quadrat $A_i$ are recorded, we can directly model
those locations as a log-linear IPP over the entire surveyed domain
$\bigcup_i A_i$. Often we do not have access to such granular data,
and only the count $N_{ik}=N_{\S_k}(A_i)$ or presence/absence
$y_{ik}$ is recorded. In such cases,
the IPP model still induces a GLM likelihood for the available summary
statistics $N_{ik}$ or $y_{ik}$, so that we can maximize likelihood
for the available data.

If the features are continuous, then for a small quadrat $A_i$ the
species count at the site is
\begin{linenomath}\begin{equation}
  N_{ik} = N_{\S_k}(A_i) \approx \text{Pois}(|A_i|\lambda_k(s_i))
  = \text{Pois}\left(|A_i|\exp\{\alpha_k + \beta_k'x(s_i)\}\right).
\end{equation}\end{linenomath}
Thus our joint IPP model induces a Poisson log-linear model for survey
count data. The probability of $y_{ik}=1$ is
\begin{linenomath}\begin{equation}
  \P(N_{ik}>0) \approx 1-\exp\{-|A_i|\lambda_k(s_i)\}
  = 1-\exp\{-e^{\alpha_k + \beta_k'x(s_i) + \log |A_i|}\},
\end{equation}\end{linenomath}
a Bernoulli GLM with complementary log-log link
\citep{mccullagh1989generalized, baddeley2010spatial}.
The complementary log-log link has
been used before to study presence-absence data in ecology
\citep[e.g.][]{yee1991generalized, royle2008hierarchical,
  lindenmayer2009nest}. If the
expected count $\eta=|A_i|\lambda_k(s_i)$ is very small then there is not
much difference between the complementary log-log link, the logistic
link, and the log link, since
\begin{linenomath}\begin{equation}
  1-\exp\{-e^\eta\}
  \approx
  \frac{   e^\eta }{ 1 + e^\eta}
  \approx
  e^\eta.
  % 1-\exp\{-|A_i|e^{\alpha_k + \beta_k'x(s)}\}
  % \approx
  % \frac{   |A_i|e^{\alpha_k+\beta_k'x(s)} }{ 1 + |A_i|e^{\alpha_k+\beta_k'x(s)}}.
  % \approx
  % |A_i|e^{\alpha_k+\beta_k'x(s)}
\end{equation}\end{linenomath}
For simplicity assume quadrat sizes are constant and
work in units where $|A_i|=1$.  When this is not the case, $\log
|A_i|$ enters as an offset in the GLM for observation $i$.

Importantly, we make no assumption that the survey quadrats $A_i$ are
distributed evenly across $\D$ in any sense.  However, our model does
assume that, given the locations of $A_i$, the {\em responses}
$y_{ik}$ for the presence-absence data are in no way impacted by
$b(s)$, the sampling bias of the presence-only data.

Informally, if the $A_i$ tend to cluster near
some population center, then we will see many presences $y_{ik}=1$
{\em and} absences $y_{ik}=0$ there, so we will not be fooled into
believing the species is more prevalent there. Because we are only
modeling the distribution of $y_{ik}$, the
presence-absence data do not suffer from selection bias even if the
geographic distribution of quadrats is very uneven.

\subsection{Target Group Background Method}

\citet{phillips2009sample} suggested another method of using many
species' presence-only data to account for sampling bias.  Using
a discretization of $\D$ into grid cells, they propose
sampling background points only from grid cells where at least one
species was sighted, guaranteeing that completely inaccessible
areas play no role in estimation.  This method, dubbed the ``target-group background'' (TGB) method, can
tackle sampling bias with only presence-only data, and
without requiring specification of its functional form.

However, the TGB method does not distinguish between
inaccessible regions and regions in which all the species are not very
prevalent. Moreover, because it samples background points equally from
all accessible grid cells, the TGB method does not adjust for biased
sampling from one accessible region relative to another. Our method
can leverage presence-absence data to directly estimate sampling bias
and predict absolute prevalence. Section~\ref{sec:euca} empirically
compares our method's out-of-sample predictive performance to several
competitors including the TGB method.

\subsection{Maximum Likelihood Estimation}\label{sec:mle}

In this section we discuss estimation of our joint model. As we will
see, maximum likelihood estimation amounts to fitting a very large
generalized linear model to all of the data.
Moreover, several familiar methods for single-species distribution modeling
amount to exactly or approximately maximizing our model's likelihood
for a specific subset of our joint data set.

Because we have various sorts of observation sites $s_i$ we introduce
notation to allow for summing over relevant subsets of them.
Let $I_\PA$ denote the set of indices $i$ for which $s_i$ are
presence-absence survey quadrats, and let $I_{\PO_k}$ denote the
indices for presence-only sites $s_i\in \S_k$. Let $n_\PA$ be the
total number of survey quadrats.

For species $k$, the log-likelihood for the presence-absence data is
\begin{linenomath}\begin{equation}\label{eq:paLikEst}
  \ell_{k,\PA}(\alpha_k,\beta_k)
  = \sum_{i\in I_\PA}
  -y_{ik}\log
  \left(1-e^{-\exp\{\alpha_k+\beta_k'x_i\}}\right)+(1-
  y_{ik})\exp\{\alpha_k+\beta_k'x_i\}.
\end{equation}\end{linenomath}
If $\P(y_i=1)$ is small for each quadrat, then $\ell_{k,PA}$ is very
close to the log-likelihood for logistic regression on
presence-absence data. In other words, applying our method to a single
presence-absence data set with no other data reduces to something very
close to presence-absence logistic regression for that species.

The log-likelihood for the presence-only data is
\begin{linenomath}\begin{align}
  \ell_{k,\PO_k}&(\alpha_k,\beta_k,\gamma_k,\delta)
  = \sum_{i\in I_{\PO_k}} \log\left(\lambda_k\cdot b_k(s_i)\right)
  - \int_{\D}\lambda_k\cdot b_k(s)\,ds\\\label{eq:poLikEst}
  &= \sum_{i\in I_{\PO_k}} \left(\alpha_k + \beta_k'x_i +
    \gamma_k + \delta'z_i\right)
  - \int_{\D}e^{ \alpha_k + \beta_k'x_i +
    \gamma_k + \delta'z_i }\,ds
% \frac{|\D|}{n_{\text{BG}}}\sum_{i\in \text{BG}} e^{\alpha_k + \beta_k'x_i +
%     \gamma_k + \delta'z_i},
\end{align}\end{linenomath}
In general we cannot evaluate the integral in~(\ref{eq:poLikEst})
exactly. As usual, we replace the integral with a weighted sum over
$n_\BG$ background sites $s_i\in\D$. For weights $w_i$, we obtain the
numerical approximation
\begin{linenomath}\begin{equation}\label{eq:numLoglik}
  \ell_{k,\PO_k}(\alpha_k,\beta_k,\gamma_k,\delta) \approx
  \sum_{i\in I_{\PO_k}} \left(\alpha_k + \beta_k'x_i +
    \gamma_k + \delta'z_i\right)
  -\sum_{i\in I_\BG} w_i e^{\alpha_k + \beta_k'x_i +
    \gamma_k + \delta'z_i},
\end{equation}\end{linenomath}
where $I_\BG$ are the indices corresponding to background sites.
In the simplest case, the background sites are sampled uniformly from
$\D$ and all the $w_i=\frac{|\D|}{n_\BG}$, but other sampling schemes
are possible \citep[for a review of techniques see
][]{renner2014point}. Popular
procedures like Maxent and presence-background logistic regression
approximately maximize~(\ref{eq:numLoglik}).

Maximizing~(\ref{eq:numLoglik}) for a single species $k$ with the
$\gamma_k+\delta'z_i$ terms included reduces to the regression
adjustment strategy discussed in Section~\ref{sec:regadj}. If we do
not include $\gamma_k+\delta'z_i$ terms (i.e., if we assume there is
no bias) we obtain the unadjusted fit (i.e. the usual fit)
to the biased presence-only intensity $\lambda_k(s)\,b_k(s)$.
% As shown in~\citet{fithian2013finite}, we can fit (\ref{eq:numLoglik}) via
% infinitely weighted logistic regression (IWLR),
% using pseudo-responses $\tilde y_{ik}$
% which are 1 if $s_i\in \T_k$ and 0 if $s_i$ is a background point. In
% IWLR, the weights $w_i$ for background points are magnified by a large
% fraction.

The presence-absence and presence-only data sets for all $m$ species
together represent $2m$ independent data sets.\footnote{Technically, the
  portion of $\T_k$ that coincides
  with survey quadrats $A_i$ is not independent of the
  presence-absence data for species $k$. We could repair this by discarding all
  presence-only and background sites occurring in survey quadrats, but
  in practice this is unnecessary because the $A_i$ represent a
  miniscule fraction of the domain.}
Maximizing likelihood
for all the data means maximizing the sum
\begin{linenomath}\begin{equation}\label{eq:fullLoglik}
  \ell(\theta) = \sum_k \ell_{k,\PA}(\alpha_k,\beta_k) + \ell_{k,\PO}(\alpha_k,\beta_k,\gamma_k,\delta),
\end{equation}\end{linenomath}
where $\theta$ represents the full complement of coefficients
\begin{linenomath}\begin{equation}
  \theta =
  (\alpha_1,\beta_1,\gamma_1,\ldots,\alpha_m,\beta_m,\gamma_m,
  \delta).
\end{equation}\end{linenomath}
With a bit of work we can massage the form of~(\ref{eq:fullLoglik})
into one large GLM in terms of a common set of $m(p+2)+r$ predictors
corresponding to the entries of $\theta$. We do so by introducing
auxiliary predictor variables $u_k$, a binary indicator that we are
predicting for species $k$, and $v$, an indicator that we are
predicting for presence-only instead of presence-absence data. In
terms of these variables, $\alpha_k$ is the coefficient for $u_k$,
$\beta_{k,j}$ for $u_k x_j$, $\gamma_{k}$ for $u_k v$, and $\delta_j$ for
$v z_j$. More details are given in Appendix~\ref{secComp}.

The result is a very large GLM with $m(p+2)+r$ total
parameters and $m(n_\BG + n_\PA)$ total observations (one per species
for each survey site and background site). Because
both the number of observations and number of parameters scale
linearly with $m$, the computational cost of standard approaches to
estimation scales as $m^3p^2(n_\BG+n_\PA)$.

For our eucalypt example, we have $m=36$ species,
$n_\BG=40,000$ background sites, $n_\PA=32,612$ survey quadrats, and
$p=38$ predictors (including interactions and nonlinear terms), so
$m^3p^2(n_\BG+n_\PA)\approx 5\times 10^{12}$. This is a very
high computational load even for modern computers.

Fortunately, there is a great deal of structure in the design matrix,
and if we exploit it properly, our computations need only scale
linearly with $m$, cutting the cost by a factor of roughly
$36^2\approx 1000$. Appendix~\ref{secComp} also details our efficient
computing scheme.

\subsection{Fitting Proportional Bias Models in R}
\newcommand{\mytilde}{\raise.17ex\hbox{$\scriptstyle\mathtt{\sim}$}}
%\newcommand{\mytilde}{\raise.17ex\hbox{$\mathtt{\sim}$}}

% For larger data sets like the one we analyze in this article, packages
% not specially tailored to exploit the block structure of the design
% may have prohibitive computation costs.
As a companion to this article, we have released an R package,
\texttt{multispeciesPP}, that can efficiently fit the models described here.
The method requires formulae for the species intensity and the
sampling bias, and carries out maximum likelihood as described in
Section~\ref{sec:mle}.  For example, the code
\begin{quote}
  \texttt{mod <- multispeciesPP(\mytilde x1 + x2, \mytilde z, PA=PA, PO=PO, BG=BG)}
\end{quote}
would fit a multispecies Poisson process model with presence-absence
data set \texttt{PA}, list of presence-only data sets \texttt{PO}, and
background data \texttt{BG}.  The R function maximizes likelihood
under the model
\begin{linenomath}\begin{align}\label{eq:logLinAssumption}
  \log \lambda_k(s_i) &= \alpha_k + \beta_{k,1} x_{i,1} + \beta_{k,2}x_{i,2}\\\label{eq:propBiasAssumption}
  \log b_k(s_i) &= \gamma_k + \delta z_i
\end{align}\end{linenomath}
and returns fitted coefficients, and predictions.

\section{Simulation}\label{sec:toysim}

Thus far, we have discussed several distinct data sources we can
bring to bear on estimating $\lambda_k(s)$, the intensity for
the $k$th species process.
A simple simulation illustrates the interplay of the different
data types.

We simulate from the model \eqref{eq:loglinbias} with covariates
$(x_1,x_2,z)$ following a trivariate
normal distribution with mean zero and covariance
\begin{linenomath}\begin{equation}
  \Cov(x_1,x_2,z) = \begin{pmatrix} 1 & 0 & 0.95 \\ 0 & 1 & 0 \\ 0.95 &
    0 & 1\end{pmatrix},
\end{equation}\end{linenomath}
and the coefficients for species 1 equal to:
\begin{linenomath}\begin{equation}
  (\alpha_1, \beta_{1,1}, \beta_{1,2}, \gamma_1, \delta) = (-2, 1, -0.5, -4, -0.3)
\end{equation}\end{linenomath}
Presence-absence data for species 1
is the most reliable reflection of $\lambda_1(s)$, but is
available only in small quantities.
Presence-only data for species 1 are
abundant, but biased, as they are sampled from the intensity
\begin{linenomath}\begin{equation}
  \lambda_1(s)\cdot b_1(s) = \alpha_1 + \beta_1' x(s) + \gamma_1 + \delta' z(s)
\end{equation}\end{linenomath}
Because $z$ is independent of $x_1$ but highly correlated with $x_2$,
a presence-only data point is mainly informative about $\beta_{1,1}$ and
$\beta_{1,2}+\delta$. Without supplementary data it carries almost
no information about $\beta_{1,2}$ itself.

If presence-only and presence-absence
data are available for many other species, then they all
contribute information helping us to precisely estimate $\delta$.
This makes species 1's presence-only data much more useful: given a
precise estimate of $\delta$ from other species' data, information
about $\beta_{1,2}+\delta$ is equivalent to information about
$\beta_{1,2}$.

Figure~\ref{fig:confEll} and the accompanying commentary
shows what each data set contributes to
estimating $\beta_{1,1}$ and $\beta_{1,2}$ by plotting the
95\% Wald confidence ellipse for each of several models.

\begin{figure}
%\begin{mdframed}%[leftmargin=.5in,rightmargin=.5in]
  \centering
  \vskip-5em
  \includegraphics[width=\textwidth]{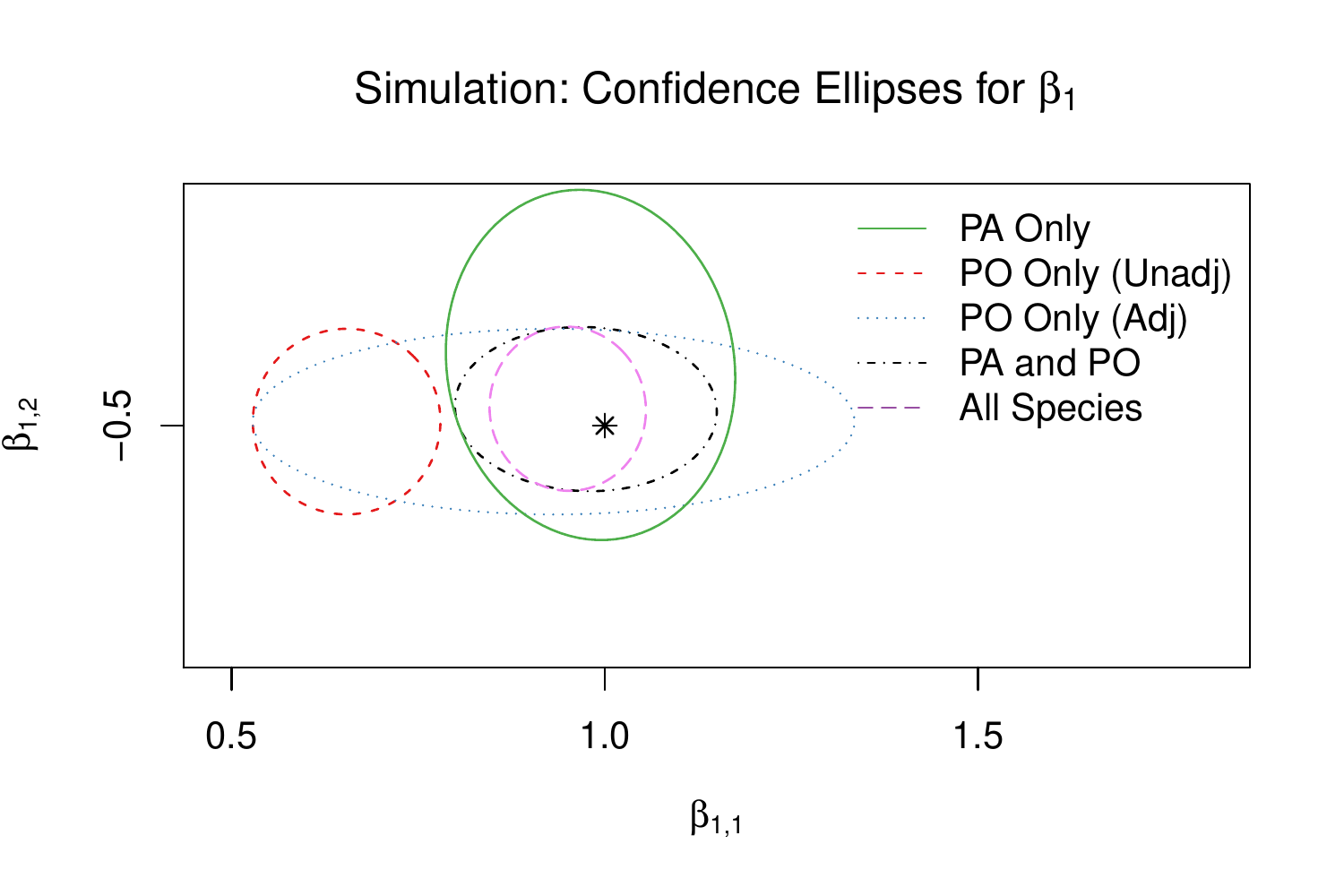}
  \vskip-2em
  \caption{95\% Wald confidence regions for $\beta_1$, the species
    distribution coefficients for species 1, obtained by
    using five different methods. The plot illustrates the precision
    and accuracy with which the coefficients are estimated by each
    method.  The black star denotes the true values
    of the parameters of interest.  The different model types are
    described below:}
  \label{fig:confEll}
    \begin{description}
    \item[PA data alone (Green):] The most straightforward method when
      PA data for species 1 is to maximize likelihood
      for it alone.  Our estimates of both coefficients are
      unbiased but less precise than they could be.
      $z$ plays no role in the PA data or our model for it, so the
      precisions for the two coordinates of $\beta_1$ are about the same.
    \item[PO data alone, no regression adjustment (Red):]
      The most common use of presence-only data is to
      maximize likelihood using only the presence-only data for
      species 1, making no adjustment for sampling bias.
      In that case, we are effectively estimating the presence-only
      intensity instead of the species intensity.  Here $x_1$ proxies
      for the confounding variable
      $z$ and $\beta_{1,2}$ is severely biased, whereas $\beta_{1,2}$
      is unaffected.
    \item[PO data alone, with regression adjustment (Blue):] We can
      address sampling bias by attempting to estimate the effect of
      the confounder $z$. Our estimates are now unbiased,
      but $\hat\beta_{1,1}$ is noisy and its interval is very wide.
      It is quite hard to tease apart the effects of $x_1$ and $z$
      given only PO data.
    \item[PA and PO data for species 1 (Black):] The PO
      data carry solid information about $\beta_{1,2}$, whereas
      the PA data carry the only usable information about $\beta_{1,1}$.
      When we combine both data sources for species 1, the precision
      of $\hat\beta_{1,2}$ roughly matches the methods using PO
      alone (blue and red), and the precision of $\hat\beta_{1,1}$
      matches the method using PA alone (green).
    \item[Pooled data for all species (Purple):] We obtain the best
      results by pooling both presence-absence and presence-only data sets
      for many different species.  Species $2,3,\ldots,m$ all
      contribute to estimating $\delta$ to high precision.  As a
      result the presence-only data for species 1 becomes much more
      useful for estimating $\beta_{1,1}$, because we know how to
      correct for the sampling bias.
    \end{description}
%\end{mdframed}
\end{figure}

\section{Eucalypt Data}\label{sec:euca}

We have just seen how the various sources of data
can work in concert to give far more precise estimates than we could
obtain from any one data set by itself.  Additionally, we evaluate our
model's performance on a  data set of 36 species of genera
{\em Eucalyptus, Corymbia}, and {\em Angophora} in southeastern Australia.

The presence-absence data consist of 32,612 sites where all the species
were surveyed, with an average of 547 presences per species.  The
species exhibit
a great deal of variability with respect to their overall abundance,
with 4 species having fewer than 20 total observations, and 8 having
more than 1000.

The presence-only data consist of 764 observations on average
per species, supplemented with 40,000 background points sampled
uniformly at random from the study region.  More information on data
sources may be found in Appendix~C. The rarest species in the
presence-only data, {\em Eucalyptus stenostoma}, has 90 observations.

We use 15 environmental covariates in
our model for the species process, allowing for nonlinear effects in 4
of them: temperature seasonality, rainfall seasonality, precipitation
in June/July/August, moisture index in the lowest quarter, and annual
precipitation overall.  Our model for the bias includes nonlinear effects for
predictors including distance to road, distance to the nearest town,
distance to the coast, ruggedness, whether the locale has extant
vegetation, and the number of presence-absence sites
nearby. Appendix~\ref{sec:moreResults} discusses the
model form in more detail.

The four panels of Figure~\ref{fig:eucapunc} contrast our model's fit for
a single species, {\em Eucalyptus punctata}, with the fit that we
would obtain by using presence-only data alone with no bias
adjustment. A satellite image of
the same region provided for comparison and orientation. The top left
panel displays the fitted intensity we obtain by modeling {\em
  E. punctata}'s presence-only data as an IPP whose intensity is
driven by environmental variables.  We obtain an estimate of the
presence-only intensity, which in this case is concentrated mostly
near Sydney and the coast.

The top right and lower left panels show our model's estimates $\hat
b_k(s)$ of the bias and $\hat \lambda_k(s)$ of the species intensity.
Unsurprisingly, distance from the coast, and from Sydney, are strong
drivers of our model's fitted sampling bias.
In the lower left panel, the intensity is shifted significantly toward
the western hinterland.

\begin{figure}
  \centering
  \includegraphics[width=1.2\textwidth]{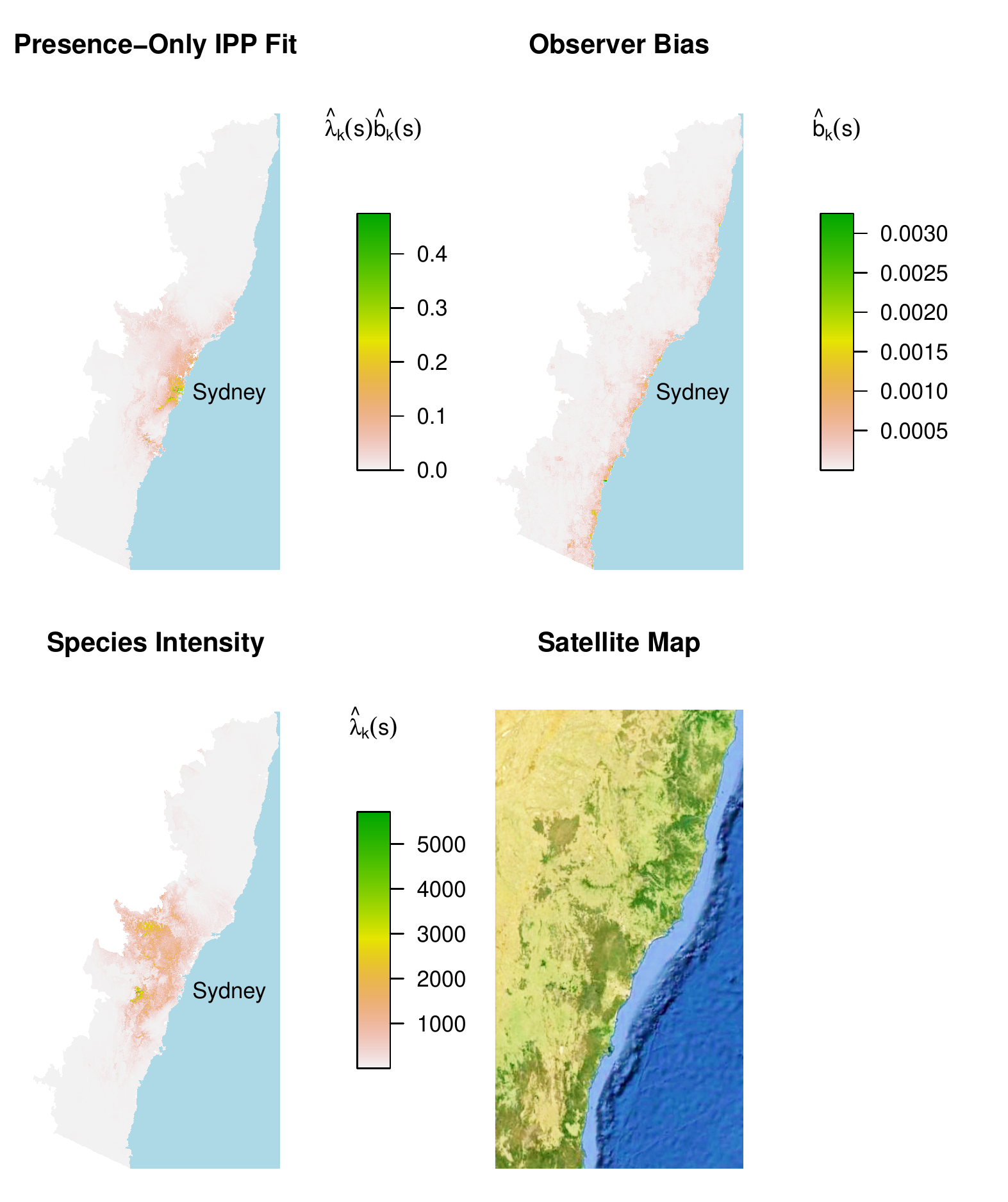}
  \caption{Model fits for {\em Eucalyptus punctata} in southeastern Australia.
    Top left panel: estimate of presence-only intensity in units of $1/\text{km}^2$,
    using presence-only data alone and making no
    adjustment for bias. Top right: fitted sampling bias $\hat b_k(s)$
    in our proportional sampling bias model. Lower left: fitted species
    intensity $\hat \lambda_k(s)$ for our model, in units of
    $1/\text{km}^2$.  Lower right: satellite image from
    Google Earth. In the presence-only data, many more trees were
    observed in near Sydney than in the western hinterland, but our
    model infers a higher intensity in the undersampled western region.}
  \label{fig:eucapunc}
\end{figure}

To evaluate our model quantitatively, we ask two questions: first, how
well do the data agree with the assumption of proportional sampling bias?
Second, do we obtain better predictions when pooling multiple
data sets across multiple species?

\subsection{Checking the Proportional Bias Assumption}\label{sec:check}

We can check the proportional bias assumption within the context of
our GLM.  To check whether the bias coefficient corresponding to some
$z_j$ should vary by species, we can estimate the same model as
before, but now allowing that coordinate of $\delta$ to vary by
species.

In terms of the large GLM described in
Section~\ref{sec:mle}, we can estimate our model as before by
augmenting the design matrix with interactions between the species
identifiers $u_k$ and the bias variable $z_j$. These variables then
have coefficients $\delta_{k,j}$. In this model the proportional bias
assumption corresponds to the null hypothesis of
no interaction effects, which we can test using standard
likelihood-based methods.

As usual, it is rather unlikely that the proportional bias assumption
--- or any other aspect of our model --- holds exactly.  Even if the
assumption holds for some true functions $\lambda_k(s)$ and $b_k(s)$,
we may still see spurious correlations when we fit a complex model
using a misspecified log-linear functional form.  Nevertheless, it is
of interest to identify whether some interactions stand out
strongly compared to the noise level, and if so how large they are.

% To test our assumption, we can estimate the model with parameters
% $\delta_k = d_k + \bar \delta$, with the constraint that
% $\sum d_k = 0$. Under this parameterization, the
% proportional bias assumption corresponds to the null hypothesis
% $H_0:\,d_k = 0$.  If we construct confidence intervals for $d_k$ that do not
% include 0, we would say that the data reject the proportional bias
% assumption.

Because of spatial autocorrelation in both the presence-absence and
presence-only data, traditional likelihood-based confidence intervals
for the interaction effects $\delta_{k,j}$ are likely to be
anticonservative, as are bootstrap intervals based on i.i.d. resampling.
To account properly for the spatial autocorrelation, we use the {\em block
bootstrap} to compute confidence intervals for the coefficients \citep{efron1993introduction}.  We
separate the landscape
into a checkerboard pattern with 261 rectangular regions with sides of length
1/3-degree of longitude and latitude (approximately 31 km $\times$ 37
km at latitude $33^\circ$ South). In each of 400 bootstrap replicates we
resample 261 whole regions with replacement.

% Even if the proportional bias assumption truly held, we might still
% reject $H_0$ above if the model is otherwise misspecified.  For
% example, suppose the bias model needs an interaction between latitude and
% another covariate of $z$ such as ruggedness, but we fail to include
% that interaction in the model.  Then if species $k$ occurs mainly in
% the Northern part of the region, and species $k'$ occurs in the South,
% it may appear to our model that $b_k$ depends more or less on
% ruggedness than $b_{k'}$ does.

\paragraph{Dependence of $\delta$ on Species:}

We test our assumption explicitly for the variable ``distance to
coast,'' which is the most important predictor of bias.
The evidence in the data regarding our assumption is somewhat mixed,
but on the whole it does not appear that the proportional bias model
fits the data perfectly.  For some species, there is sufficient
evidence to reject $H_0$.

Figure~\ref{fig:interact} shows the 95\% bootstrap confidence interval
for the idiosyncratic sampling bias of {\em Eucalyptus punctata}, as a
function of distance to coast.  We see that, even after accounting for the
overall bias that affects the other 35 species, we still have too many coastal
presence-only observations of {\em punctata}.  This could be linked to
the fact that the {\em punctata} data are concentrated near Sydney,
which is more heavily populated than other coastal regions, but with
many confounding factors at play it is hard to know.
Appendix~\ref{sec:moreResults} has more detailed results for more
species.

If interactions like these are strong, we can allow some of the
coordinates of $\delta$ to vary by $k$ and others not.  There is a
bias-variance tradeoff, however, since the proportional-bias
assumption is what allows us to share information across species.
We will see in Section~\ref{sec:cv} that even when the model is an
imperfect fit, it can nevertheless substantially improve
predictive performance on held-out presence-absence data.

\begin{figure}
  \centering
  \includegraphics[width=.9\textwidth]{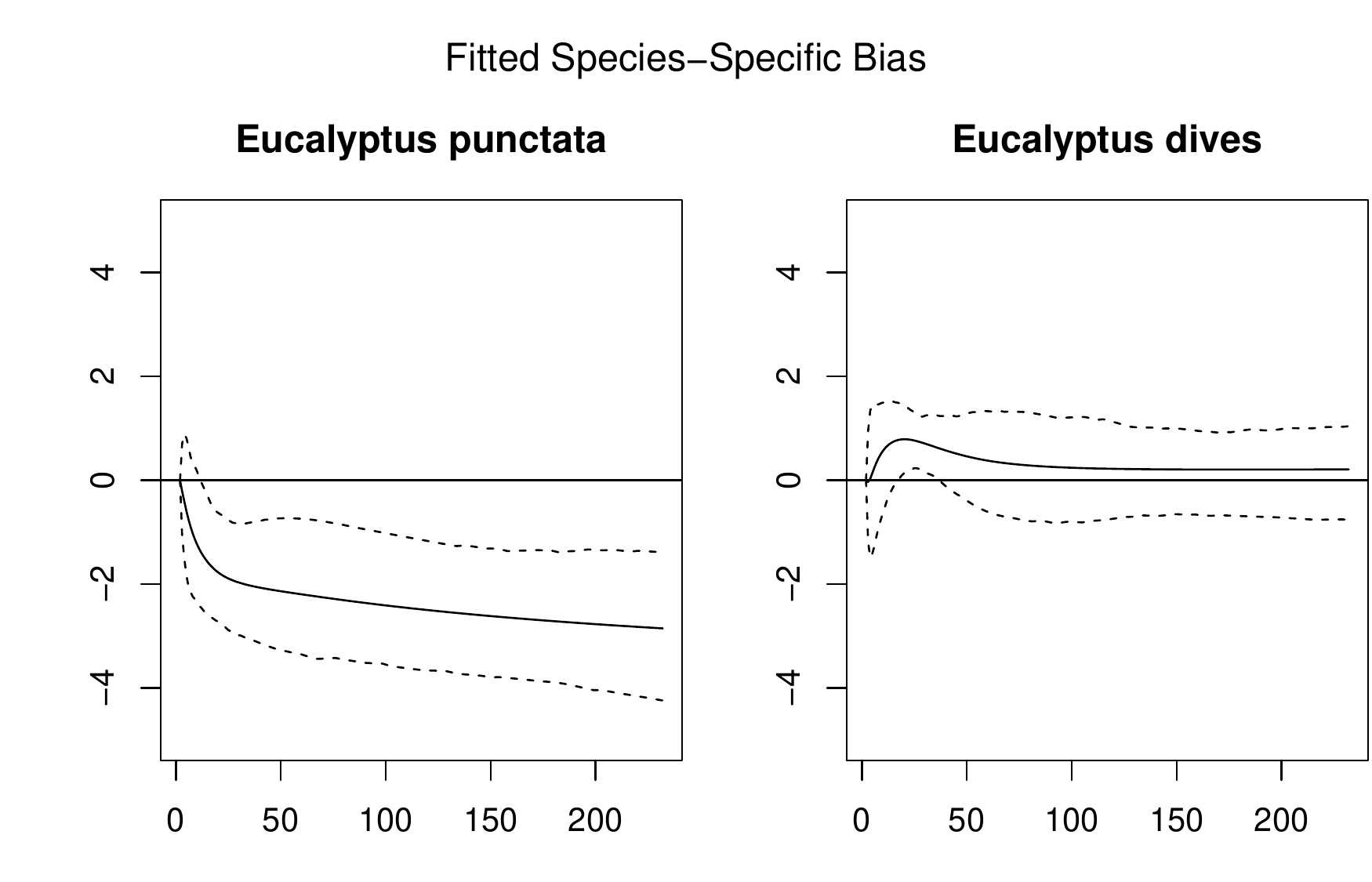}
  \caption{Idiosyncratic sampling bias for {\em E. punctata} and {\em
      E. dives} as a function of distance to coast in km.
    The dashed lines show 95\%
    block-bootstrap confidence intervals.  It appears that after adjusting
    for the bias $\delta'z(s)$ that is shared across all species,
    there is some residual bias left over for {\em punctata}. By
    contrast, for {\em E. dives}, there is no
    significant interaction.  Even though the proportional
    sampling bias model is misspecified for {\em E. punctata}, it
    still substantially improves out-of-sample predictive accuracy, as
    we will see in Section~\ref{sec:cv}. The corresponding curves for
    all the species can be
    found in Appendix~\ref{sec:moreResults}.}
  \label{fig:interact}
\end{figure}

\paragraph{Dependence of $\gamma$ on Species:}

By default, our model allows $\gamma$ to vary by species, but we need
not always do so.  In fact, if we assumed $\gamma$ does not vary by
species, then we would only need joint presence-absence and presence-only
data for one species to obtain an estimate for $\gamma$.  Therefore, we
could estimate abundance (and therefore presence
probabilities) for {\em every} species given presence-absence and
presence-only data for a single species and presence-only data for
every other species.

Define {\em relative sampling effort} as the ratio
\begin{linenomath}\begin{equation}
  \rho_k = \frac{\exp\{\gamma_k\}}{\displaystyle\min_{k'=1}^m \exp\{\gamma_{k'}\}},
\end{equation}\end{linenomath}
so that $\rho_k=1$ for all $k$ if and only if the $\gamma_k$ are all equal.

Figure~\ref{fig:samplingEffort} shows our model's estimates $\hat\rho_k$,
plotted against the total number of presence-absence observations.
For the eucalypt data it appears that the assumption of a
common $\gamma$ for every species is probably not reasonable.
It appears the presence-only intercept $\gamma$
varies systematically by species, with effort being substantially
higher for the rarer species. Thus, the data appear to support our
decision to allow $\gamma$ to vary by species.

\begin{figure}
  \centering
  \includegraphics[width=.6\textwidth]{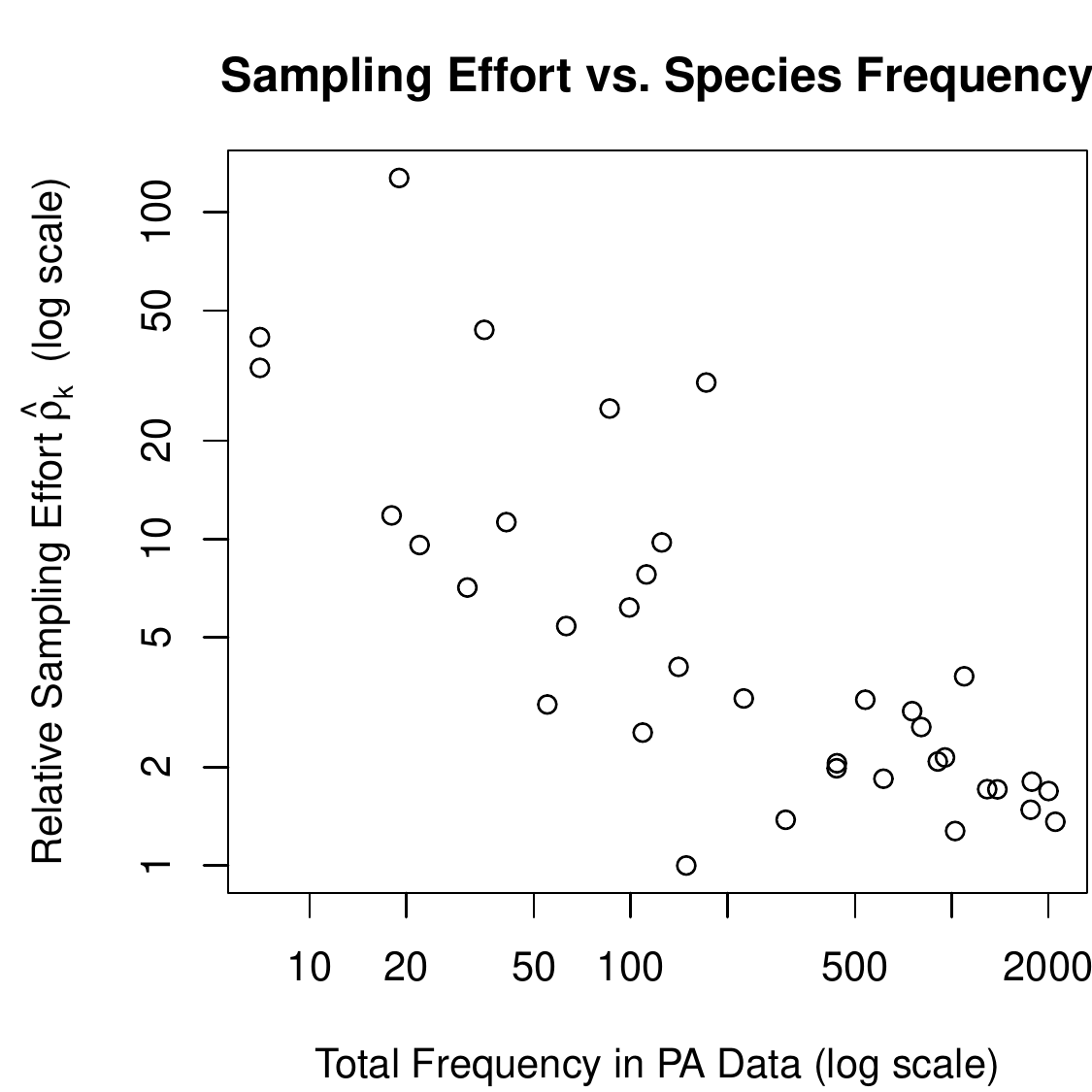}
  \caption{Our model's estimate of relative sampling effort
    $\rho_k$, plotted
    versus the total abundance of each species, with each variable
    plotted on a log scale.  It appears that more
    effort is made to sample rare species.}
  \label{fig:samplingEffort}
\end{figure}

\subsection{Predictive Evaluation of the Model}\label{sec:cv}

Our goal in pooling data was to supplement the presence-absence data
for a given species with multiple other more abundant
sources of data, to allow for
more efficient estimation of the species intensity $\lambda_k(s)$ and
its coefficients.  One measure of our success
is whether this data pooling actually improves predictive
performance on held-out presence-absence data.

For comparison, we also estimate our joint
model using a) both the presence-only and presence-absence data for
species $k$, and b) presence-only and presence-absence data for all 36
species combined.

Note that in all three cases we are estimating the exact same joint
model with three nested data sets:
\begin{description}
\item[PA data alone for species $k$:] The most natural
  competitor to our method is
  to fit the Bernoulli complementary log-log GLM model with the same
  predictors, but only on species $k$'s presence-absence data.
  This is a special case of the joint method, for which only
  presence-absence data are available for species $k$.
\item[PA and PO data for species $k$:] Augmenting
  the presence-absence data with presence-only data for the same
  species improves our coefficient estimates for environmental
  variables that are independent of sampling bias.\\
  When there is no presence-absence data, we are fitting the thinned Poisson process model
  to PO data alone. This is
  regression-adjusted analysis of PO data, discussed in Section~\ref{sec:regadj}.
\item[Pooled data for all species:] Using data for all species gives
  better estimates of the predictors that are badly confounded by
  sampling bias.
\end{description}

In addition we introduce two more competitors that use presence-only
data alone:
\begin{description}
\item[PO data alone for species $k$, unadjusted for bias:] Using
  species $k$'s presence-only data alone, and ignoring sampling bias,
  is the most common method for analyzing presence-only data.  It
  estimates the presence-only intensity, and then makes predictions as
  though that were the same as the species intensity.  This method can
  suffer dramatically from bias.
\item[PO data for all species, using the TGB method:] We implement the
  TGB method with pixel size 9 arc seconds (the resolution level of
  our covariates).
\end{description}

Our evaluation method effectively
treats the presence-absence data as a ``gold standard,'' unaffected by
bias.  This point of view may not always be reasonable, but
eucalypts are relatively large and hard for surveyors to miss, so the
presence-absence data probably do reflect the true presence or absence
of trees in their respective quadrats, notwithstanding identification
errors.

We emphasize that we are comparing the different methods with respect
to their performance on held-out presence-absence data and not on
held-out presence-only data. This distinction is important, because
our goal is to reconstruct the species intensity and not the
presence-only intensity. All three methods train on the
same amount of presence-absence data for species $k$. The
data-pooling methods can only beat the simpler method if the other
data sets carry useful information about the species intensity of
species $k$, and if our joint model effectively processes that
information without biasing our estimate too badly.

We then use ten-fold block cross-validation to evaluate each method with
respect to its predictive log-likelihood.  Using the same rectangular
regions as in Section~\ref{sec:check}, we randomly assign the 261 whole regions to
ten folds, with each fold containing 26 random regions and the one
left-over region excluded.  Figure~\ref{fig:blockCV} shows one training-test
split used for our procedure.  Importantly, {\em all} data taken from the
test region --- presence-absence, presence-only, and background --- is
held out of the training set.

The gains from data pooling are greatest when the presence-absence data
for a species of particular interest (say, species $k$) are either
scarce or non-existent.  To emulate
estimation with presence-absence data sets ranging from scarce to
abundant, we further downsampled the presence-absence training data
for species $k$.

We fit all the models with a ridge penalty on all of the coefficients except the intercepts
$\alpha$ and $\gamma$.  That is, we minimize
\begin{linenomath}\begin{equation}
  \ell(\alpha,\beta,\gamma,\delta) + \frac{\nu}{2}\|\beta\|_2^2 +
  \frac{\nu}{2}\|\delta\|_2^2,
\end{equation}\end{linenomath}
with penalty multiplier $\nu = 100$.  Penalizing the
coefficients in this way is known as {\em regularization}, and it
allows for efficient estimation of parameters in complex models.
For more details see e.g. \citet{ESL}.

Figures~\ref{fig:cvEucapunc} and \ref{fig:cvEucadive} show the results of block
cross-validation for two species in the data set: {\em Eucalyptus
  punctata} and {\em Eucalyptus
  dives}.  Results for the other species are qualitatively similar and
can be found in Appendix~\ref{sec:moreResults}.  We evaluate the
various methods according to two metrics of predictive performance:
predictive log-likelihood (left panel) and area under the predictive
ROC curve, averaged over the ten test folds (AUC, right panel).
\citet{lawson2014prevalence} contrast {\em prevalence-dependent}
metrics like log-likelihood, which measure the accuracy of absolute
out-of-sample presence probabilities, with {\em prevalence-independent} metrics
like AUC, which depend only on the ordering of predictions.

Doing well in predictive log-likelihood requires a good estimate of
the intercept $\alpha_k$ --- i.e., of the {\em absolute} intensity $\lambda_k(s)$.
Because $\alpha_k$ is confounded with
$\gamma_k$ in presence-only data, and because $\gamma_k$ varies by
species, the two data-pooling methods cannot estimate {\em absolute}
intensities without a little presence-absence data from species $k$.
By contrast,
AUC only depends on estimates of {\em relative} intensity
$\frac{\lambda_k(s)}{\Lambda_k(\D)}$, which is
invariant to $\hat\alpha_k$ and can be estimated with {\em no}
presence-absence data for species $k$.  Estimates without any
presence-absence data for species $k$ are shown above the label ``0''
on the horizontal axis.

As we have seen in Figure~\ref{fig:eucapunc}, {\em E. punctata}
suffers dramatically from sampling bias because Sydney, the largest
city, lies on the eastern edge of its habitable zone.  As a result,
the unadjusted presence-only method performs very poorly compared to
the methods that account for bias.  By contrast, the habitable zone of {\em
  E. dives} lies mainly in the western part of the study region where
the sampling bias function $\log b_k$ has a much gentler gradient.
As a result, the unadjusted presence-only analysis does relatively
well. The method that pools across all 36 species does even better:
its AUC using {\em none} of {\em E. punctata}'s presence-absence data
(and only the presence-absence data for the other 35 species) is
indistinguishable from its AUC using {\em all} of the presence-absence
data. See Appendix~\ref{sec:moreResults}
for the corresponding plots for all species.

Table~\ref{tab:cvRes} compares the four best methods using a moderate value, 1000, for the number of non-missing presence-absence sites. Our method pooling presence-absence and presence-only data for all species performs well consistently, coming within 0.01 of the best method for all but one species. Interestingly, the TGB method performs second best despite its having no access to the presence-absence data.

\begin{figure}
  \centering
  \includegraphics[width=.6\textwidth]{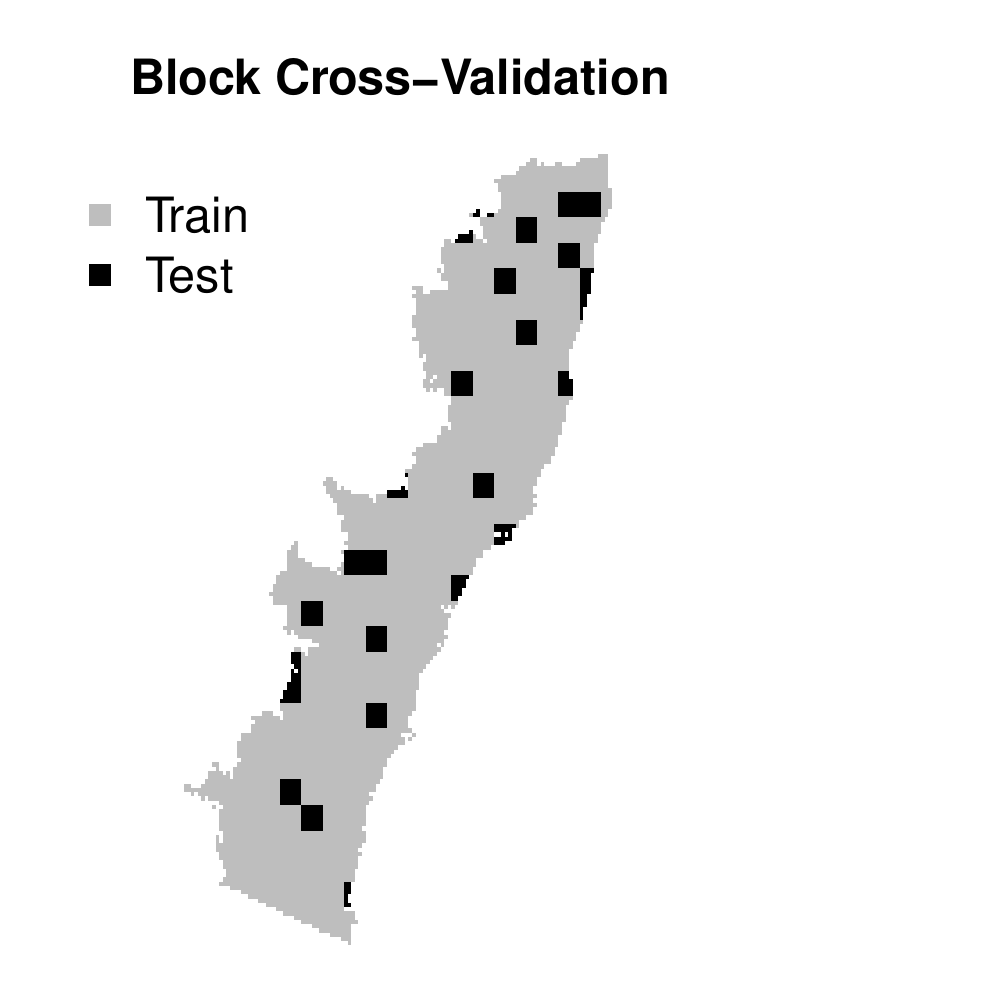}
  \caption{Depiction of our block cross-validation scheme for the eucalypt data.  Entire rectangular blocks are sampled together to help account for spatial autocorrelation.}
  \label{fig:blockCV}
\end{figure}

\begin{figure}
  \begin{center}
  \includegraphics[width=\textwidth]{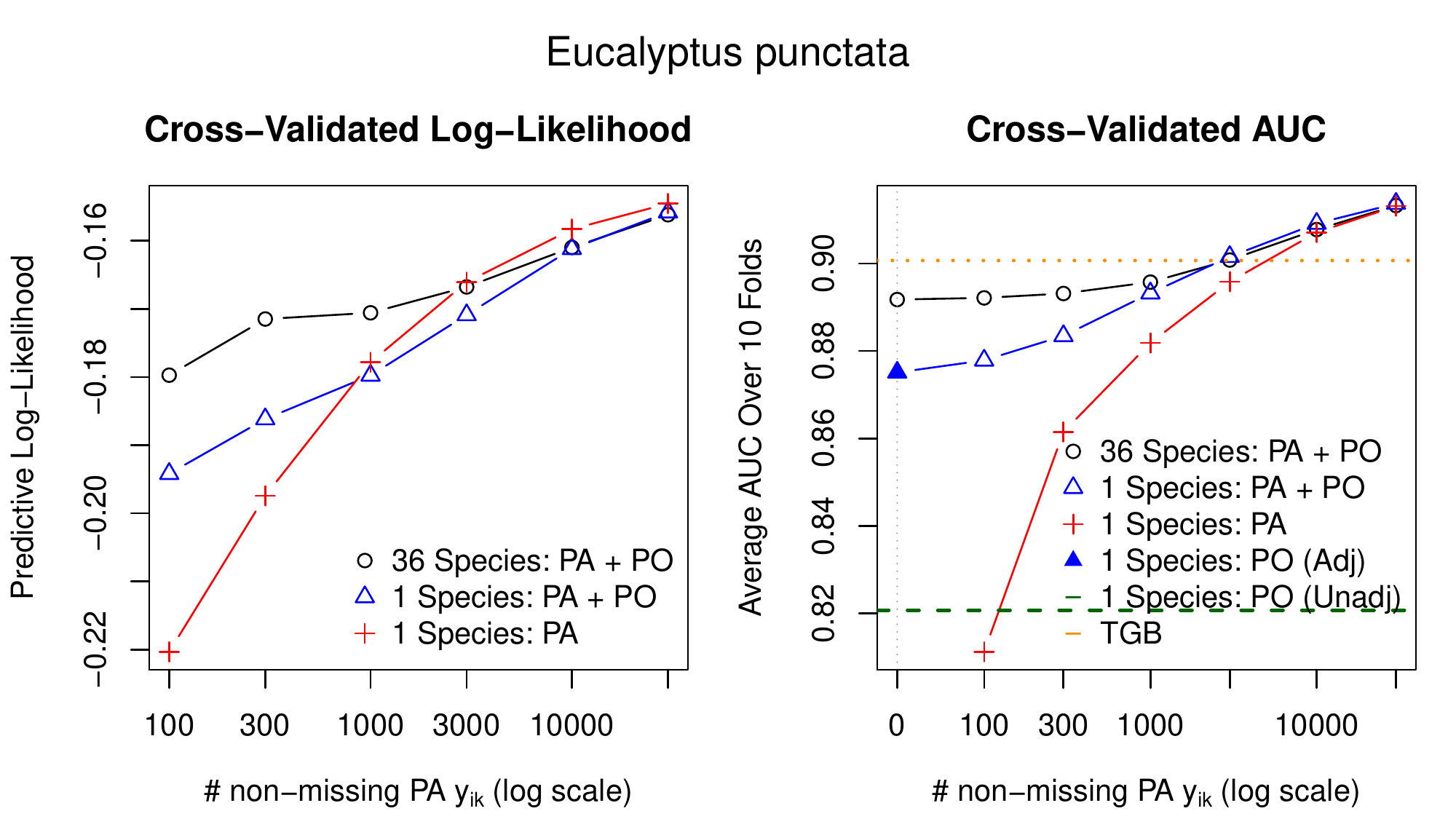}
  \caption{Block cross-validated log-likelihood and AUC for
   {\em E. punctata} (higher is better). Pooling data from other sources gives a
    substantial boost to predictive performance when the
    presence-absence data set is small, but only when we make an
    adjustment for the bias.}
  \label{fig:cvEucapunc}
  \end{center}
  In the right panel, the leftmost blue
  triangle (``1 species: PA+PO'' with no PA data),
  we are fitting the thinned IPP model
  to PO data alone. This is the regression adjustment strategy
  discussed in Section~\ref{sec:regadj}. Note that using
  presence-only data without any adjustment for bias performs quite
  poorly compared to the other methods.  Because the habitable zone for {\em
    E. punctata} includes Sydney as well as more inaccessible
  regions to its west, ignoring the sampling bias can wreak
  havoc on our estimates.
\end{figure}

\begin{figure}
  \begin{center}
  \includegraphics[width=\textwidth]{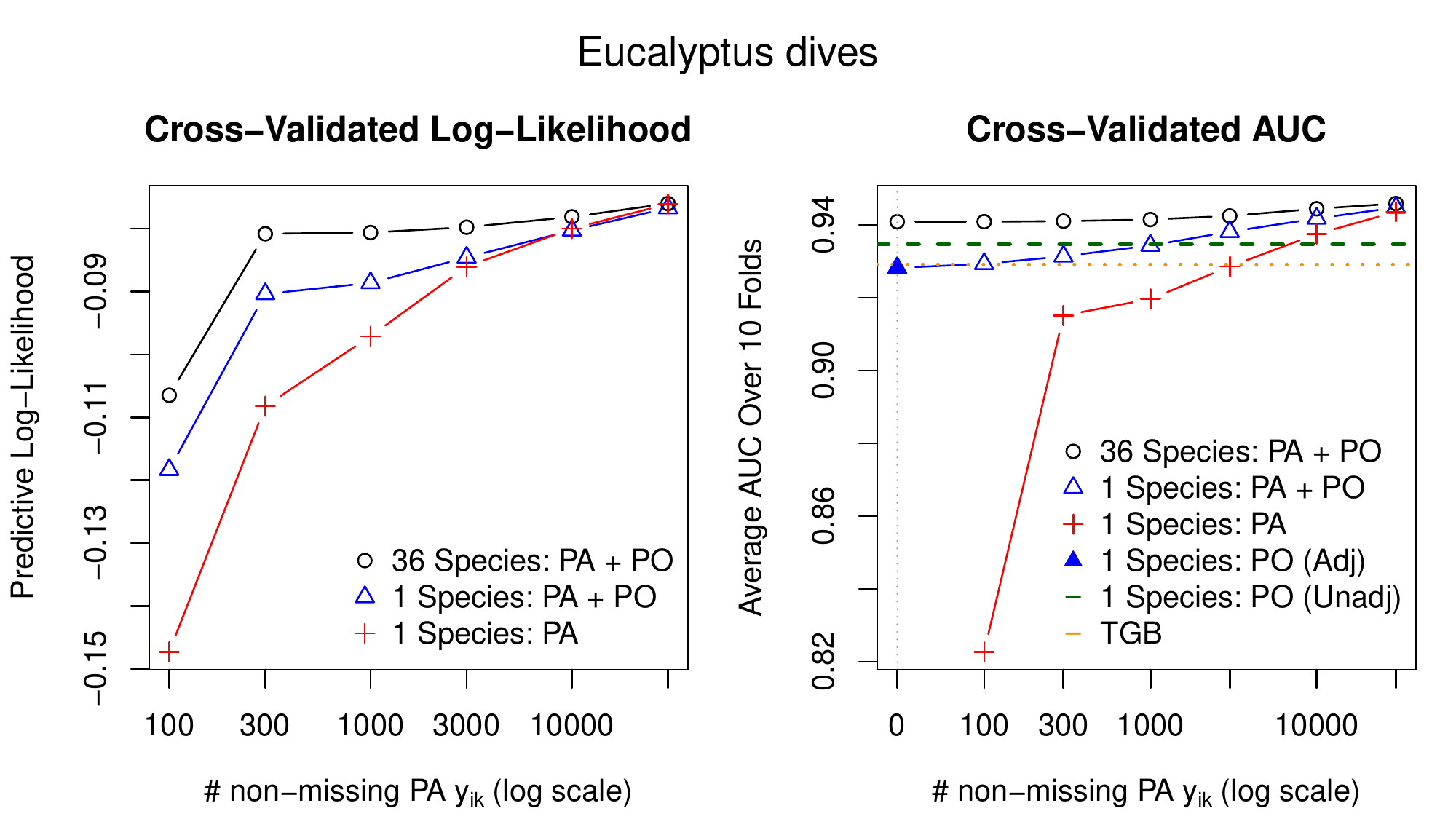}
  \caption{Block cross-validated log-likelihood and cross-valid AUC for the species
    {\em E. dives} (higher is better).  Pooling data from other sources gives a
    substantial boost to predictive performance when the
    presence-absence data set is small.  Because {\em E. dives} occurs in the
    southwestern part of the study region, where the bias function has a
    relatively gentle gradient, the sampling bias plays a less vital
    role.}
  \label{fig:cvEucadive}
\end{center}
    In the right panel, the leftmost blue
    triangle (``1 species: PA+PO'' with no PA data),
    we are fitting the thinned IPP model
    to PO data alone. This is the regression adjustment strategy
    discussed in Section~\ref{sec:regadj}.
\end{figure}

% latex table generated in R 3.0.2 by xtable 1.7-1 package
% Tue May 27 23:16:11 2014
\begin{table}[ht]\label{tab:cvRes}
\centering
\begin{tabular}{l|llll}
  \hline
 & PA Only & PA + PO & PA + PO & TGB \\
 & 1 Species & 1 Species & 36 Species & 36 Species \\
  \hline
A. bakeri & 0.893 & 0.915 & \textbf{0.932} & \textbf{0.933} \\
  C. eximia & 0.921 & \textbf{0.947} & \textbf{0.952} & \textbf{0.952} \\
  C. maculata & \textbf{0.783} & \textbf{0.778} & \textbf{0.785} & 0.742 \\
  E. agglomerata & 0.801 & \textbf{0.834} & 0.820 & 0.808 \\
  E. blaxlandii & 0.904 & 0.934 & \textbf{0.944} & 0.934 \\
  E. cypellocarpa & \textbf{0.861} & 0.852 & \textbf{0.867} & 0.825 \\
  E. dalrympleana (S) & 0.873 & 0.910 & \textbf{0.926} & \textbf{0.931} \\
  E. deanei & 0.811 & 0.855 & \textbf{0.906} & 0.894 \\
  E. delegatensis & 0.971 & 0.971 & \textbf{0.981} & \textbf{0.982} \\
  E. dives & 0.920 & \textbf{0.934} & \textbf{0.941} & 0.929 \\
  E. fastigata & 0.905 & 0.900 & \textbf{0.916} & \textbf{0.907} \\
  E. fraxinoides & 0.920 & 0.935 & \textbf{0.963} & \textbf{0.963} \\
  E. moluccana & 0.881 & \textbf{0.909} & \textbf{0.911} & 0.881 \\
  E. obliqua & 0.870 & \textbf{0.914} & \textbf{0.918} & 0.906 \\
  E. pauciflora & 0.874 & 0.897 & \textbf{0.928} & \textbf{0.928} \\
  E. pilularis & \textbf{0.807} & \textbf{0.807} & \textbf{0.805} & \textbf{0.811} \\
  E. piperita & \textbf{0.889} & 0.844 & \textbf{0.886} & 0.871 \\
  E. punctata & 0.882 & \textbf{0.893} & \textbf{0.896} & \textbf{0.901} \\
  E. quadrangulata & \textbf{0.835} & \textbf{0.843} & \textbf{0.840} & 0.823 \\
  E. robusta & 0.878 & 0.883 & \textbf{0.892} & \textbf{0.894} \\
  E. rossii & \textbf{0.957} & \textbf{0.966} & \textbf{0.965} & \textbf{0.962} \\
  E. sieberi & 0.857 & 0.813 & \textbf{0.881} & \textbf{0.875} \\
  E. tricarpa & \textbf{0.969} & \textbf{0.970} & \textbf{0.971} & \textbf{0.965} \\
   \hline
\end{tabular}
\caption{AUC cross-validation results for all species with at least 100 presence-absence data points. The first three methods are evaluated with 1000 non-missing presence absence data points for the species under study. In each row, numbers are bolded for methods coming within 0.01 of the best method. Our method pooling presence-absence and presence-only data for all species performs well consistently, coming within 0.01 of the best method for all but one species.}
\end{table}

\section{Discussion}\label{sec:discussion}

We have proposed a unifying Poisson process model that allows for
joint analysis of presence-absence and presence-only data from many
species.  By sharing information, we can obtain more precise and
reliable estimates of the species intensity than we could obtain from
either data set by itself.

Moreover, we have seen in Section~\ref{sec:euca} that the proportional
bias can be a reasonable fit for some real ecological data sets.  In
this data set, and we suspect in many others, sampling bias can have a
major effect on fitted intensities if not appropriately accounted
for.

\subsection{Benefits of Data Pooling}

Throughout we have focused mainly on the way that pooling
presence-absence and presence-only data from
many species can help address selection bias. Even when
selection bias is not a major concern, data pooling can still be
beneficial.

In the simplest case, presence-absence data can be fruitfully
supplemented by more abundant presence-only data from the same
species. In Figure~\ref{fig:cvEucadive}, we see that the
presence-only data for {\em E. dives} is not very biased, as evidenced
by the good performance of the unadjusted fit. In this case, combining
the presence-absence data with presence-only data still led to a
substantial improvement in predictive performance, and combining with
data from other species helped even more.
In other cases, we may have presence-only data for many species but
no presence-absence data. In that case, our method still provides a
means for pooling data to estimate $\delta$ more efficiently.

\subsection{Common Misspecifications of the IPP Model}\label{sec:misspec}

Aside from the proportional bias assumption, we should be mindful of
several other sources of misspecification. The most obvious is that
our log-linear functional form is almost certainly incorrect in any
given case.  Three others that
merit special consideration are spatial autocorrelation
in the data, biased detection of presence-absence data, and spatial
errors in environmental covariates and point observations.

\paragraph{Spatial Autocorrelation:} The Poisson process model assumes
that, given the
covariates for a given site, an individual is no more or less likely to
occur simply because there is another individual nearby.
In ecological data, this assumption is rather tenuous;
for example, trees of the same species often occur together in
stands; or, different species may compete with each
other for resources. \cite{renner2013equivalence} discuss
goodness-of-fit checks and present empirical evidence against the
Poisson assumption. For a more general discussion of alternatives to
the Poisson process model, see \cite{cressie1993,gaetan2009spatial}.

Similarly, for systematic survey data, we should proceed with caution
in modeling count data as Poisson, because actual counts may be
overdispersed due to autocorrelation within a quadrat, or correlated
with counts for nearby sites because of longer-range autocorrelation.
When autocorrelation is present, nominal standard errors
computed under the Poisson assumption can be much too small, as can
i.i.d. cross-validation estimates of prediction error or
i.i.d. bootstrap standard errors.
Resampling methods such as the bootstrap or cross-validation can be
made much more robust to autocorrelation if they resample whole blocks
at a time \citep{efron1993introduction}, and in Section~\ref{sec:euca}
we use the block bootstrap and block cross-validation to analyze our
eucalypt data set. Discussion of alternative block bootstrap
procedures and choosing block size may be found in
\citet{hall1995blocking,nordman2007optimal,guan2007thinned}.

\paragraph{Imperfect Detection:} Even in presence-absence and other
systematic survey data,
surveyors may not have the time or resources to exhaustively survey
a given quadrat, and thus some organisms may be missed in the
surveys.

Suppose, for example, that an organism at $s$
is detected by surveyors with probability $q(s)$. Then the count $y$
in quadrat $A$ centered at $s$ is not distributed as
$\text{Pois}(\lambda(s)|A|)$, but rather as
$\text{Pois}(q(s)\lambda(s)|A|)$. If $q(s)$ is
constant, all our estimates of $\alpha_k$ will be biased downward by
exactly $\log q$. This would bias estimates of abundance but not the
estimated species
distribution, which depends only on $\hat\beta_k$.

If $q(s)$ is a non-constant function of $s$ --- e.g., if non-detection
is a bigger problem in heavily forested sites --- then we may incur bias
for both $\alpha_k$ {\em and} $\beta_k$.  If sites are visited
repeatedly, then under some assumptions an estimate of non-detection
may be obtained, by methods discussed in
e.g. \citet{royle2003estimating, dorazio2012predicting}. Estimates of
detection probability can sometimes be obtained without repeat observations
under stronger modeling assumptions \citep{lele2012dealing,solymos2012conditional}.

Non-detection in
presence-absence data is largely analogous to the sampling bias
problem for presence-only data, and we could in principle model and
adjust for it using similar methods to the ones we propose for
addressing biased presence-only data.

\paragraph{Spatial Errors}

Opportunistic presence-only data may also suffer from errors in the
recorded locations of point observations. Similarly, environmental
covariates are often measured at a relatively coarse scale, in which case
the covariates attributed to point $s_i$ may be inaccurate. If
important environmental covariates fluctuate on a fine scale compared
to the scale of these errors, the errors may lead to attenuated
effect size estimates \citep[see e.g.][]{graham2008influence}.
\citet{hefley2013correction} propose methods to correct for spatial
errors in presence-only records.

A similar issue can arise in the analysis of presence-absence or count
data, when we use the centroid of a
presence-absence quadrat as a proxy for the integral
$\int_{A_i}\lambda(s)\,ds$, which may not be appropriate if
the variables fluctuate on a fine scale relative to quadrat size.
In such cases, it is especially helpful to record point locations
within quadrats rather than recording only presence-absence or count
data summarized at the quadrat level.

\subsection{Extensions}

As discussed elsewhere, there are many useful ways to extend GLM
fitting procedures.  GAMs, gradient-boosted trees, and other forms of
regularization on model parameters are all immediate
extensions of the approach we have outlined here.  Like other methods,
our method's results on a given data set will depend
on making good choices regarding featurization and regularization.

% In specific data sets spatial autocorrelation,
% imperfect detection or spatial errors may be important to address.
% When the baseline versions of single-species IPP or presence-absence
% models require modification to address these issues,
% our method will typically require similar modifications. Addressing
% every issue is outside of our scope here, but we believe
% the literature addressing these issues for
% single-species data should provide guidance for how to proceed.

Finally, in our approach we are forced to assume a
functional form for the sampling bias, and if our model is
wrong, we will not account correctly for the sampling bias.
Studies quantifying patterns of sampling bias in relation to
spatial covariates are currently scarce, but could help to justify a
more accurate model of sampling bias than one based on intuitive
selection of covariates, as applied here. Nonetheless,
in future work we plan to investigate models that treat the sampling
bias nonparametrically, imposing no assumptions on its functional form.

\section*{Acknowledgements}

Survey data were sourced from the NSW Office of Environment and Heritage’s (OEH) Atlas of NSW Wildlife, which holds data from a number of custodians. Data obtained July 2013.
Many thanks to Philip Gleeson, OEH, for help with understanding the
database and for checking quarantined records for us. And to
Christopher Simpson, OEH, for making the distance to roads
layer. William Fithian was supported by National Science Foundation VIGRE grant DMS-0502385. Jane Elith
was funded by Australian Research Council grant FT0991640. Trevor
Hastie was partially supported by grant DMS-1007719 from the National
Science Foundation, and grant RO1-EB001988-15 from the National
Institutes of Health. Finally, we are very grateful to Trevor Hefley,
Geert Aarts, and our editors, for their very thorough and
helpful comments which greatly improved our manuscript.

\bibliographystyle{plainnat}
\bibliography{biblio}

\begin{thebibliography}{38}
\providecommand{\natexlab}[1]{#1}
\providecommand{\url}[1]{\texttt{#1}}
\expandafter\ifx\csname urlstyle\endcsname\relax
  \providecommand{\doi}[1]{doi: #1}\else
  \providecommand{\doi}{doi: \begingroup \urlstyle{rm}\Url}\fi

\bibitem[Aarts et~al.(2012)Aarts, Fieberg, and Matthiopoulos]{AartsIPP}
G.~Aarts, J.~Fieberg, and J.~Matthiopoulos.
\newblock Comparative interpretation of count, presence--absence and point
  methods for species distribution models.
\newblock \emph{Methods in Ecology and Evolution}, 3\penalty0 (1):\penalty0
  177--187, 2012.

\bibitem[Baddeley et~al.(2010)Baddeley, Berman, Fisher, Hardegen, Milne,
  Schuhmacher, Shah, and Turner]{baddeley2010spatial}
A~Baddeley, M~Berman, NI~Fisher, A~Hardegen, RK~Milne, D~Schuhmacher, R~Shah,
  and R~Turner.
\newblock Spatial logistic regression and change-of-support in poisson point
  processes.
\newblock \emph{Electronic Journal of Statistics}, 4:\penalty0 1151--1201,
  2010.

\bibitem[Chakraborty et~al.(2011)Chakraborty, Gelfand, Wilson, Latimer, and
  Silander]{chakraborty2011point}
Avishek Chakraborty, Alan~E Gelfand, Adam~M Wilson, Andrew~M Latimer, and
  John~A Silander.
\newblock Point pattern modelling for degraded presence-only data over large
  regions.
\newblock \emph{Journal of the Royal Statistical Society: Series C (Applied
  Statistics)}, 60\penalty0 (5):\penalty0 757--776, 2011.

\bibitem[Cressie(1993)]{cressie1993}
N.A.C. Cressie.
\newblock \emph{Statistics for Spatial Data, revised edition}, volume 928.
\newblock Wiley, New York, 1993.

\bibitem[Dorazio(2012)]{dorazio2012predicting}
Robert~M Dorazio.
\newblock Predicting the geographic distribution of a species from
  presence-only data subject to detection errors.
\newblock \emph{Biometrics}, 68\penalty0 (4):\penalty0 1303--1312, 2012.

\bibitem[Dud{\i}k et~al.(2005)Dud{\i}k, Schapire, and
  Phillips]{dudik2005correcting}
Miroslav Dud{\i}k, Robert~E Schapire, and Steven~J Phillips.
\newblock Correcting sample selection bias in maximum entropy density
  estimation.
\newblock \emph{Advances in neural information processing systems},
  17:\penalty0 323--330, 2005.

\bibitem[Efron and Tibshirani(1993)]{efron1993introduction}
Bradley Efron and Robert Tibshirani.
\newblock \emph{An introduction to the bootstrap}, volume~57.
\newblock CRC press, 1993.

\bibitem[Elith et~al.(2011)Elith, Phillips, Hastie, Dud{\'\i}k, Chee, and
  Yates]{elith2011statistical}
J.~Elith, S.J. Phillips, T.~Hastie, M.~Dud{\'\i}k, Y.E. Chee, and C.J. Yates.
\newblock A statistical explanation of maxent for ecologists.
\newblock \emph{Diversity and Distributions}, 2011.

\bibitem[Fithian and Hastie(2013)]{fithian2013finite}
William Fithian and Trevor Hastie.
\newblock Finite-sample equivalence in statistical models for presence-only
  data.
\newblock \emph{The Annals of Applied Statistics}, 7\penalty0 (4):\penalty0
  1917--1939, 2013.

\bibitem[Gaetan and Guyon(2009)]{gaetan2009spatial}
C.~Gaetan and X.~Guyon.
\newblock \emph{Spatial statistics and modeling}.
\newblock Springer Verlag, 2009.

\bibitem[Giraud et~al.(2014)Giraud, Calenge, and
  Julliard]{giraud2014capitalising}
Christophe Giraud, Cl{\'e}ment Calenge, and Romain Julliard.
\newblock Capitalising on opportunistic data for monitoring biodiversity.
\newblock \emph{arXiv preprint arXiv:1407.2432}, 2014.

\bibitem[Graham et~al.(2008)Graham, Elith, Hijmans, Guisan, Townsend~Peterson,
  and Loiselle]{graham2008influence}
Catherine~H Graham, Jane Elith, Robert~J Hijmans, Antoine Guisan,
  A~Townsend~Peterson, and Bette~A Loiselle.
\newblock The influence of spatial errors in species occurrence data used in
  distribution models.
\newblock \emph{Journal of Applied Ecology}, 45\penalty0 (1):\penalty0
  239--247, 2008.

\bibitem[Guan and Loh(2007)]{guan2007thinned}
Yongtao Guan and Ji~Meng Loh.
\newblock A thinned block bootstrap variance estimation procedure for
  inhomogeneous spatial point patterns.
\newblock \emph{Journal of the American Statistical Association}, 102\penalty0
  (480):\penalty0 1377--1386, 2007.

\bibitem[Hall et~al.(1995)Hall, Horowitz, and Jing]{hall1995blocking}
Peter Hall, Joel~L Horowitz, and Bing-Yi Jing.
\newblock On blocking rules for the bootstrap with dependent data.
\newblock \emph{Biometrika}, 82\penalty0 (3):\penalty0 561--574, 1995.

\bibitem[Hastie et~al.(2009)Hastie, Tibshirani, and Friedman]{ESL}
T.~Hastie, R.~Tibshirani, and J.~Friedman.
\newblock \emph{The elements of statistical learning}.
\newblock Springer Series in Statistics, 2009.

\bibitem[Hastie and Fithian(2013)]{hastie2013inference}
Trevor Hastie and Will Fithian.
\newblock Inference from presence-only data; the ongoing controversy.
\newblock \emph{Ecography}, 36\penalty0 (8):\penalty0 864--867, 2013.

\bibitem[Hefley et~al.(2013{\natexlab{a}})Hefley, Baasch, Tyre, and
  Blankenship]{hefley2013correction}
Trevor~J Hefley, David~M Baasch, Andrew~J Tyre, and Erin~E Blankenship.
\newblock Correction of location errors for presence-only species distribution
  models.
\newblock \emph{Methods in Ecology and Evolution}, 2013{\natexlab{a}}.

\bibitem[Hefley et~al.(2013{\natexlab{b}})Hefley, Tyre, Baasch, and
  Blankenship]{hefley2013nondetection}
Trevor~J Hefley, Andrew~J Tyre, David~M Baasch, and Erin~E Blankenship.
\newblock Nondetection sampling bias in marked presence-only data.
\newblock \emph{Ecology and Evolution}, 2013{\natexlab{b}}.

\bibitem[Lawson et~al.(2014)Lawson, Hodgson, Wilson, and
  Richards]{lawson2014prevalence}
Callum~R Lawson, Jenny~A Hodgson, Robert~J Wilson, and Shane~A Richards.
\newblock Prevalence, thresholds and the performance of presence--absence
  models.
\newblock \emph{Methods in Ecology and Evolution}, 5\penalty0 (1):\penalty0
  54--64, 2014.

\bibitem[Lehmann and Casella(1998)]{lehmann1998theory}
Erich~Leo Lehmann and George Casella.
\newblock \emph{Theory of point estimation}, volume~31.
\newblock Springer, 1998.

\bibitem[Lele and Keim(2006)]{lele2006weighted}
Subhash~R Lele and Jonah~L Keim.
\newblock Weighted distributions and estimation of resource selection
  probability functions.
\newblock \emph{Ecology}, 87\penalty0 (12):\penalty0 3021--3028, 2006.

\bibitem[Lele et~al.(2012)Lele, Moreno, and Bayne]{lele2012dealing}
Subhash~R Lele, Monica Moreno, and Erin Bayne.
\newblock Dealing with detection error in site occupancy surveys: what can we
  do with a single survey?
\newblock \emph{Journal of Plant Ecology}, 5\penalty0 (1):\penalty0 22--31,
  2012.

\bibitem[Lindenmayer et~al.(2009)Lindenmayer, Welsh, Donnelly, Crane, Michael,
  Macgregor, McBurney, Montague-Drake, and Gibbons]{lindenmayer2009nest}
David~B Lindenmayer, Alan Welsh, Christine Donnelly, Mason Crane, Damian
  Michael, Christopher Macgregor, Lachlan McBurney, Rebecca Montague-Drake, and
  Philip Gibbons.
\newblock Are nest boxes a viable alternative source of cavities for
  hollow-dependent animals? long-term monitoring of nest box occupancy, pest
  use and attrition.
\newblock \emph{Biological conservation}, 142\penalty0 (1):\penalty0 33--42,
  2009.

\bibitem[McCullagh and Nelder(1989)]{mccullagh1989generalized}
P~McCullagh and John~A Nelder.
\newblock \emph{Generalized Linear Models}, volume~37.
\newblock CRC Press, 1989.

\bibitem[Nordman et~al.(2007)Nordman, Lahiri, and Fridley]{nordman2007optimal}
Daniel~J Nordman, Soumendra~N Lahiri, and Brooke~L Fridley.
\newblock Optimal block size for variance estimation by a spatial block
  bootstrap method.
\newblock \emph{Sankhy{\=a}: The Indian Journal of Statistics}, pages 468--493,
  2007.

\bibitem[Pearce and Boyce(2006)]{pearce2006modelling}
Jennie~L Pearce and Mark~S Boyce.
\newblock Modelling distribution and abundance with presence-only data.
\newblock \emph{Journal of Applied Ecology}, 43\penalty0 (3):\penalty0
  405--412, 2006.

\bibitem[Phillips et~al.(2009)Phillips, Dud{\'\i}k, Elith, Graham, Lehmann,
  Leathwick, and Ferrier]{phillips2009sample}
Steven~J Phillips, Miroslav Dud{\'\i}k, Jane Elith, Catherine~H Graham, Anthony
  Lehmann, John Leathwick, and Simon Ferrier.
\newblock Sample selection bias and presence-only distribution models:
  implications for background and pseudo-absence data.
\newblock \emph{Ecological Applications}, 19\penalty0 (1):\penalty0 181--197,
  2009.

\bibitem[Renner and Warton(2013)]{renner2013equivalence}
Ian~W Renner and David~I Warton.
\newblock Equivalence of maxent and poisson point process models for species
  distribution modeling in ecology.
\newblock \emph{Biometrics}, 2013.

\bibitem[Renner et~al.(2014)Renner, Baddeley, Elith, Fithian, Hastie, Phillips,
  Popovic, and Warton]{renner2014point}
Ian~W Renner, Adrian Baddeley, Jane Elith, William Fithian, Trevor Hastie,
  Steven Phillips, Gordana Popovic, and David~I Warton.
\newblock Point process models for presence-only analysis --- a review.
\newblock \emph{Methods in Ecology and Evolution}, 2014.

\bibitem[Royle and Dorazio(2008)]{royle2008hierarchical}
J~Andrew Royle and Robert~M Dorazio.
\newblock \emph{Hierarchical modeling and inference in ecology: the analysis of
  data from populations, metapopulations and communities}.
\newblock Academic Press, 2008.

\bibitem[Royle and Nichols(2003)]{royle2003estimating}
J~Andrew Royle and James~D Nichols.
\newblock Estimating abundance from repeated presence-absence data or point
  counts.
\newblock \emph{Ecology}, 84\penalty0 (3):\penalty0 777--790, 2003.

\bibitem[Royle et~al.(2012)Royle, Chandler, Yackulic, and
  Nichols]{royle2012likelihood}
J~Andrew Royle, Richard~B Chandler, Charles Yackulic, and James~D Nichols.
\newblock Likelihood analysis of species occurrence probability from
  presence-only data for modelling species distributions.
\newblock \emph{Methods in Ecology and Evolution}, 3\penalty0 (3):\penalty0
  545--554, 2012.

\bibitem[S{\'o}lymos et~al.(2012)S{\'o}lymos, Lele, and
  Bayne]{solymos2012conditional}
P{\'e}ter S{\'o}lymos, Subhash Lele, and Erin Bayne.
\newblock Conditional likelihood approach for analyzing single visit abundance
  survey data in the presence of zero inflation and detection error.
\newblock \emph{Environmetrics}, 23\penalty0 (2):\penalty0 197--205, 2012.

\bibitem[Ward et~al.(2009)Ward, Hastie, Barry, Elith, and
  Leathwick]{ward2009em}
G.~Ward, T.~Hastie, S.~Barry, J.~Elith, and J.R. Leathwick.
\newblock Presence-only data and the em algorithm.
\newblock \emph{Biometrics}, 65\penalty0 (2):\penalty0 554--563, 2009.

\bibitem[Warton et~al.(2013)Warton, Renner, and Ramp]{warton2013model}
David~I Warton, Ian~W Renner, and Daniel Ramp.
\newblock Model-based control of observer bias for the analysis of
  presence-only data in ecology.
\newblock \emph{PloS one}, 8\penalty0 (11):\penalty0 e79168, 2013.

\bibitem[Warton and Shepherd(2010)]{WartonIPP}
D.I. Warton and L.C. Shepherd.
\newblock Poisson point process models solve the "pseudo-absence problem" for
  presence-only data in ecology.
\newblock \emph{The Annals of Applied Statistics}, 4\penalty0 (3):\penalty0
  1383--1402, 2010.

\bibitem[Yee and Mitchell(1991)]{yee1991generalized}
Thomas~W Yee and Neil~D Mitchell.
\newblock Generalized additive models in plant ecology.
\newblock \emph{Journal of vegetation science}, 2\penalty0 (5):\penalty0
  587--602, 1991.

\bibitem[Zaniewski et~al.(2002)Zaniewski, Lehmann, and
  Overton]{zaniewski2002predicting}
A~Elizabeth Zaniewski, Anthony Lehmann, and Jacob~McC Overton.
\newblock Predicting species spatial distributions using presence-only data: a
  case study of native new zealand ferns.
\newblock \emph{Ecological modelling}, 157\penalty0 (2):\penalty0 261--280,
  2002.

\end{thebibliography}

\newpage

\begin{appendix}

  \section{Maximum Likelihood Estimation as a Joint GLM}
  \label{secComp}

  Recall that maximizing likelihood for the full data set means
  maximizing
  \begin{linenomath}\begin{equation}\label{eq:fullApx}
    \ell(\theta) = \sum_k \ell_{k,\PA}(\alpha_k,\beta_k) + \ell_{k,\PO}(\alpha_k,\beta_k,\gamma_k,\delta),
  \end{equation}\end{linenomath}
  where
  \begin{linenomath}\begin{align}\label{eq:paApx}
    \ell_{k,\PA}(\alpha_k,\beta_k)
    &= \sum_{i\in I_\PA}
    -y_{ik}\log
    \left(1-e^{-\exp\{\alpha_k+\beta_k'x_i\}}\right)+(1-
    y_{ik})\exp\{\alpha_k+\beta_k'x_i\}\\\label{eq:poApx}
    \ell_{k,\PO_k}(\alpha_k,\beta_k,\gamma_k,\delta) &\approx
    \sum_{i\in I_{\PO_k}} \left(\alpha_k + \beta_k'x_i +
      \gamma_k + \delta'z_i\right)
    -\sum_{i\in I_\BG} w_i e^{\alpha_k + \beta_k'x_i +
      \gamma_k + \delta'z_i}
  \end{align}\end{linenomath}
  In this section we discuss how to massage~(\ref{eq:fullApx}) into a
  large GLM in terms of a common set of $m(p+2)+r$ predictors and
  coefficients. For the moment, we ignore the sum over $I_{\PO_k}$
  in~(\ref{eq:poApx}) and deal with the other two sums. The sum
  in~(\ref{eq:paApx}) is the log-likelihood for a Bernoulli GLM with
  complementary log-log link and the sum over $I_\BG$
  in~(\ref{eq:poApx}) is the log-likelihood for a weighted Poisson GLM
  with log link.

  Note that at each survey site we have $m$
  presence-absence observations, one for every species. Similarly, we
  will introduce one ``dummy'' response $y_{ik}=0$ for each species
  $k$ at each background site $i$, for $m(n_\PA+n_\BG)$ total
  observations. For observation $ik$, introduce auxiliary indicator
  variables
  \begin{linenomath}\begin{align}
    u_{ik_1,k_2} &= \left\{\begin{matrix} 1 & k_1 = k_2\\
        0 & \text{otherwise}\\\end{matrix}\right.\\
    v_{ik} &= \left\{\begin{matrix} 1 & i\in I_\BG\\
        0 & \text{otherwise}\\\end{matrix}\right.
  \end{align}\end{linenomath}
  The variable $u_k$ allows parameters to vary by species. For
  example, $\alpha_k$ is the coefficient for $u_k$ and $\beta_{k,j}$
  is the coefficient for the interaction $x_j u_k$. The variable $v$
  gives us bias terms that apply only to terms in the presence-only
  likelihood. Thus $\gamma_k$ is the coefficient for $u_k v$ and
  $\delta_j$ is the coefficient for $z_j v$.

  For example, the linear
  predictor for count or presence-absence for species $k$ at a
  survey site with predictors $x$ and $z$ is
  \begin{linenomath}\begin{equation}
    \alpha_k + \beta_k'x_i = \sum_{1\leq h \leq m} \left(\alpha_h u_{ik,h} +
      \beta_{h}' x_i u_{ik,h} + \gamma_h u_{ik,h} v_{ik}\right) + \delta' z_{i} v_{ik},
  \end{equation}\end{linenomath}
  using $v=0$ because we are predicting for presence-absence data.

  To check the proportional bias assumption for variable $z_j$ ---
  that is, to check the assumption that $\delta_j$ should be the same
  for every species --- we can augment the model with interactions
  $z_{j*k} = u_kz_j$ for each $k$, and test the hypothesis that each
  of those variables has no effect on the regression.

  Let $X_\PA$ denote the $n_\PA \times p$ matrix with all $x$
  variables for all the survey sites, and let $X_\BG$ and $Z_\BG$
  denote all the $x$ and $z$ variables for all the background sites.
  Then if
  \begin{linenomath}\begin{equation}
    X = \begin{pmatrix} 1 & X_{\text{PA}} & 0\\ 1 & X_{\text{BG}}
      & 1\end{pmatrix},\;\;\;
    Z = \begin{pmatrix} 0 \\ Z_{\text{BG}}\end{pmatrix},
  \end{equation}\end{linenomath}
  our likelihood is a large weighted GLM with
  $m(n_{\text{PA}}+n_{\text{BG}})$ observations and overall design matrix
  \begin{linenomath}\begin{equation}
    \mathbb{X} =
    \begin{pmatrix}
      X & 0 & \cdots & 0 & Z\\
      0 & X & \cdots & 0 & Z\\
      \vdots & \vdots & \ddots & \vdots & \vdots\\
      0 & 0 & \cdots & X & Z\\
    \end{pmatrix}, \quad \text{ and coefficients } \quad
    \theta = \begin{pmatrix} \theta_1 \\ \vdots \\ \theta_m \\
      \delta \end{pmatrix}.
  \end{equation}\end{linenomath}
  The weights are $w_i$ for rows corresponding to background site $i$,
  and 1 for presence-absence sites.
  Note that the response family and link function are different for
  different rows.

  Turning to the sums over $I_{\PO_k}$
  in~(\ref{eq:poApx}), note that
  they are linear in the coefficients, so all $k$ sums can be combined
  to obtain a single linear term of the form $\theta'M$.
  All the parameters may be estimated simultaneously
  via a slight modification of iterative reweighted least squares that
  takes into account the $m$ linear terms.

  \subsection{Iterative Reweighted Least Squares Using Block Structure}
  Let $n=n_{\text{PA}}+n_{\text{BG}}$.
  $\bbX$ has $mn$ rows and $m(p+2) + r$ columns.
  In principle, we could form the matrix $\bbX$ and use standard GLM
  software to fit the model, but we would pay a
  very high computational price for estimating multiple species at a
  time.

  The main computational bottleneck in each iteration is solving a
  large weighted linear least-squares problem with $mn$ equations (one
  per species per site) and $m(p+2) + r$
  unknowns.  The update for step $t$ requires solving a weighted
  linear least-squares problem with row weights
  $W^{(t)}=\diag\left(w^{(t)}\right)$ and working responses $u^{(t)}$:
  \begin{linenomath}\begin{equation}\label{eq:wtdLS}
    \min_\theta \left\| W^{(t)} \left(\bbX
        \theta - u^{(t)}\right)\right\|_2^2.
  \end{equation}\end{linenomath}
  Solving a completely general problem of the form~\eqref{eq:wtdLS} would
  require $\mathcal{O}(m^3np^2+mnr^2)$ floating point operations.
  Fortunately, we can store and compute
  much more cheaply if we exploit the special block structure of $\bbX$.

  Our computational scheme relies heavily on the following well-known
  and highly useful lemma:
  \begin{lemma}[Partitioned Least Squares]\label{lem:partLS}
    Consider the least-squares problem
    \begin{linenomath}\begin{equation}\label{eq:lemLS}
      \min_{v} \left\|(A\;\; B) \binom{v_1}{v_2} - c\right\|_2^2.
    \end{equation}\end{linenomath}
    Let $B._{A}$ represent the matrix $B$ with each column
    orthogonalized with respect to the column space of $A$.  Then for
    $v^*$ solving \eqref{eq:lemLS} we have
    \begin{linenomath}\begin{equation}
      B._A'B._A v_2^* = B._{A}'c = B._{A}'c._{A}.
    \end{equation}\end{linenomath}

    That is, the least-squares coefficients for $B$ may be obtained by
    first regressing the columns of $B$ on $A$, then regressing $c$ on the
    residuals.
  \end{lemma}
  \begin{proof}
    Let $M$ be least-squares coefficients for regression of $B$ on
    $A$; that is,
    \begin{linenomath}\begin{equation}
      B = AM + B._A
    \end{equation}\end{linenomath}
    Then,~(\ref{eq:lemLS}) is equivalent to the least-squares problem
    \begin{linenomath}\begin{equation}\label{eq:lemLS2}
      \min_{\bar v} \left\|(A\;\; B._A) \binom{\bar v_1}{\bar v_2} - c\right\|_2^2.
    \end{equation}\end{linenomath}
    To see why, note that
    \begin{linenomath}\begin{equation}
      A\bar v_1 + B._A \bar v_2 = A\left(\bar v_1 - M v_2\right) +
      B\bar v_2
    \end{equation}\end{linenomath}
    so solutions to~(\ref{eq:lemLS}) and~(\ref{eq:lemLS2}) are in
    direct correspondence with one another, with $v_2 = \bar v_2$.

    Moreover, because the two blocks in~(\ref{eq:lemLS2}) are
    orthogonal to each other, we can solve the problem by separately
    regressing $c$ on $A$ and on $B._A$ to obtain $\bar v_1^*$ and
    $v_2^*=\bar v_2^*$.
  \end{proof}
  Our proof implies further that having obtained $M$ and
  $v_2^*$, we can compute $v_1^* = \bar v_1^* - Mv_2^*$.

  \subsection{Least Squares with Block Structure}

  Suppressing the $t$ superscript, we need to solve a least squares
  problem with design matrix $W\bbX$ and response vector $u$.  Writing
  \begin{linenomath}\begin{equation}
%    W = \begin{pmatrix} W_1 & & & \\ & W_2 & & \\ &  & \ddots & \\ & &
%      & W_m \end{pmatrix},
    W = \begin{pmatrix} W_1 & &  \\ &  \ddots & \\ &
      & W_m \end{pmatrix},
  \end{equation}\end{linenomath}
  we have
  \begin{linenomath}\begin{align}\label{eq:blockWX}
    W\bbX =
    \begin{pmatrix} W_1 X  & &  & W_1 Z\\  &  \ddots & & \vdots\\ & &
      W_m X & W_m Z\end{pmatrix} &=
    \begin{pmatrix} X_1  & &  & Z_1\\  &  \ddots & & \vdots\\ & &
      X_m & Z_m\end{pmatrix}.
  \end{align}\end{linenomath}
  Let $\theta_1^*, \ldots, \theta_{m+1}^*$ be the blocks of least-squares
  coefficients corresponding to the column blocks
  in~(\ref{eq:blockWX}).  Writing $W\bbX = (\cX\;\cZ)$,
  Lemma~\ref{lem:partLS} means that given $\widetilde \cZ = \cZ._\cX$,
  we can efficiently solve for the coefficients $\theta_{m+1}$ by
  solving the $r\times r$ system
  \begin{linenomath}\begin{equation}\label{eq:rByr}
    \widetilde\cZ'\widetilde \cZ\theta_{m+1}^* =
  \widetilde\cZ' u
  \end{equation}\end{linenomath}

  Because $\cX$ is block diagonal, the $k$th row block of
  $\widetilde \cZ$ is $\widetilde Z_k=Z_k._{X_k}$; that is,
  orthogonalizing $\cZ$
  with respect to $\cX$ is equivalent to orthogonalizing each $Z_k$
  independently with respect to the corresponding $X_k$.
  After computing a single QR decomposition of $X_k$,
  we compute and store the least-squares coefficients
  $\bar\theta_k$ and $\Gamma_k$ from regressing $u_k$ and $Z_k$ on
  $X_k$.  Having done this we can also compute the residuals
  $\widetilde Z_k$ cheaply.

  To obtain $\theta_{m+1}^*$ in the end, we need only keep a running
  tally of the quantities appearing in~(\ref{eq:rByr}),
  \begin{linenomath}\begin{equation}
    \widetilde \cZ'\widetilde \cZ = \sum_k \widetilde Z_k'\widetilde
    Z_k,\quad \text{ and } \quad
    \widetilde \cZ'u = \sum_k \widetilde Z_k'u_k,
  \end{equation}\end{linenomath}
  and solving~(\ref{eq:rByr}) gives $\theta_{m+1}^*$.
  Now, per Lemma~(\ref{eq:lemLS}), we can reconstruct all of
  $\theta^*$ if we retain the least-squares coefficients of $u$ and
  $Z_k$ on $X_k$ at every step.  Algorithm~\ref{algo:block} gives the
  full details of the procedure.

  Most of the computational will typically be spent computing the QR
  decompositions of the blocks $X_k$.  Each QR decomposition requires
  $\mathcal{O}(np^2)$ operations, so that $\mathcal{O}(mnp^2)$ total
  operations are required for this step.  Computing $\widetilde
  \cZ'\widetilde \cZ$ requires $\mathcal{O}(mnr^2)$ operations.  Thus
  our method requires $\mathcal{O}(mn(p^2+r^2))$ operations, compared
  to $\mathcal{O}(m^3np^2+mnr^2)$ required for the naive
  method.  For $m=36$ species with $p \approx r$, for example,
  our method does roughly 650 times less work than the naive
  approach.

  Our method is also lightweight with respect to its storage costs.
  After one block's computation is completed in the first for loop of
  Algorithm~\ref{algo:block}, we do not need to store $u_k, Z_k, X_k$,
  or its QR decomposition.  We need only store the $p(r+1)$
  least-squares coefficients from each step.

  \begin{algorithm}
    \SetAlgoLined
    $A \gets 0_{r\times r}$\;
    $b \gets 0_r$\;
    \For{$k=1,\ldots,m$}{
      Compute QR decomposition of $X_k$\;
      Regress $Z_k$ on $X_k$ to obtain  $Z_k = X_k \Gamma_k + \widetilde Z_k$\;
      $A \gets A + \widetilde Z_k'\widetilde Z_k$\;
      Regress $u_k$ on $X_k$ to obtain  $u_k = X_k \bar\theta_k
      + \tilde u_k$\;
      $b \gets b + \widetilde Z_k'\tilde u_k$\;
    }
    Solve $A\theta_{m+1}^* = b$ for $\theta_{m+1}^*$\;
    \For{$k=1,\ldots,m$}{
      $\theta_k^* \gets \bar\theta_k - \Gamma_k\theta_{m+1}^*$\;
    }
    \caption{Efficient Least-Squares Using Block Structure of $W\bbX$}
    \label{algo:block}
  \end{algorithm}

\newpage

\section{Results of Eucalypt Study in More Detail}\label{sec:moreResults}

In Section~\ref{sec:check} we examined the assumption of a
proportional sampling bias, and discussed how to check this assumption
based on the data.  We checked whether distance-to-coast, an
important bias covariate, had a {\em species-specific} effect on the
sampling bias for the individual species {\em E. punctata}.
The data suggested there was an effect.  Figure~\ref{fig:allInter}
shows the analogous fitted curve for the 35 other species in the data
set.  As we see, many of these species exhibit an effect that is
either very nearly or not quite distinguishable from no effect.

We also computed cross-validated predictive performance on
log-likelihood and AUC in
Figures~\ref{fig:allLoglik}--~\ref{fig:allAUC}.  Some of
the rarest species only appeared on one or two of the geographic
blocks, so we exclude them from the cross-validation results.

For the cross-validated models, the R formula used to model the
species distribution is
\begin{quote}
\begin{verbatim}
~predict(bc04.basis, newx = bc04)
+ predict(rsea.basis, newx = rsea)
+ predict(bc33.basis, newx = bc33)
+ predict(bc12.basis, newx = bc12)
+ predict(rjja.basis, newx = rjja)
+ bc02 + bc05 + bc14 + bc21
+ bc32 + mvbf + rugg + factor(subs)
+ twmd + twmx
\end{verbatim}
\end{quote}
Descriptions of the environmental variables can be found in
Appendix~C.  Terms like \texttt{predict(var.basis, newx=var)}
appear when we have used a customized natural spline basis for the
variable \texttt{var}.

The R formula used to model the bias is
\begin{quote}
\begin{verbatim}
~predict(survey.bg.basis,newx=survey.bg)
+ predict(ld2coast.basis,newx=ld2coast)
+ predict(alongCoast.basis,newx=alongCoast)
+ predict(ld2r.basis,newx=ld2r)
+ predict(d2t.basis,newx=d2t)
+ predict(rugg.basis,newx=rugg)
+ xveg
\end{verbatim}
\end{quote}
The variable \texttt{survey.bg} is a geographic variable corresponding
to the logarithm of how many presence-absence survey sites are located
in a grid cell.  It is meant to proxy for the log-frequency of ecologist
visits to locations near a site $s\in \D$.  Note that the locations of
presence-absence survey sites are not modeled as random in our model,
so we are not using the response variable twice by doing this.

The variable \texttt{ld2coast} is the logarithm of 1 km plus the
distance to coast, and \texttt{ld2r} is the logarithm of 1m plus the
distance to the nearest road.  \texttt{d2t} is the distance to the
nearest town.  \texttt{alongCoast} is a projection of geographic
location in a direction running parallel to the coast; it is largest in the
northeastern part of the study region and smallest in the southwestern
part.  \texttt{rugg} is ruggedness of terrain, and \texttt{xveg} is
a binary indicator of whether a location has extant vegetation (e.g.,
it would be 1 in a forest and 0 in a wheat field).

\begin{figure}
  \centering
  \includegraphics[width=\textwidth]{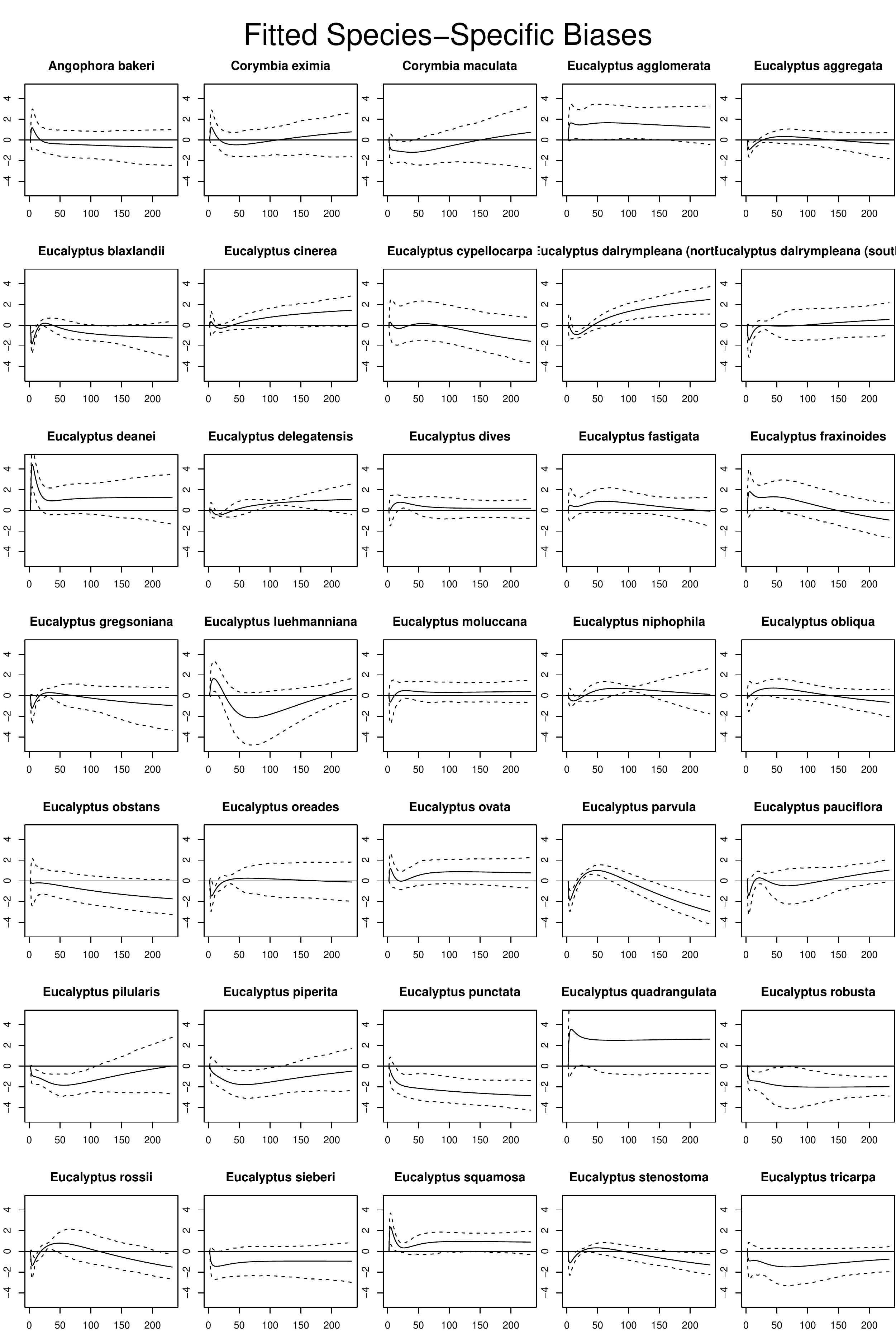}
  \caption{Bootstrap confidence intervals for the species-specific
    effect of distance-to-coast on log-sampling bias, for each of the 35
    species other than {\em Eucalyptus parramattensis}, whose data set
    is too small for the block bootstrap.  Some of the confidence intervals
    exclude zero for a significant effect.}
  \label{fig:allInter}
\end{figure}

\begin{figure}
  \centering
  \includegraphics[width=\textwidth]{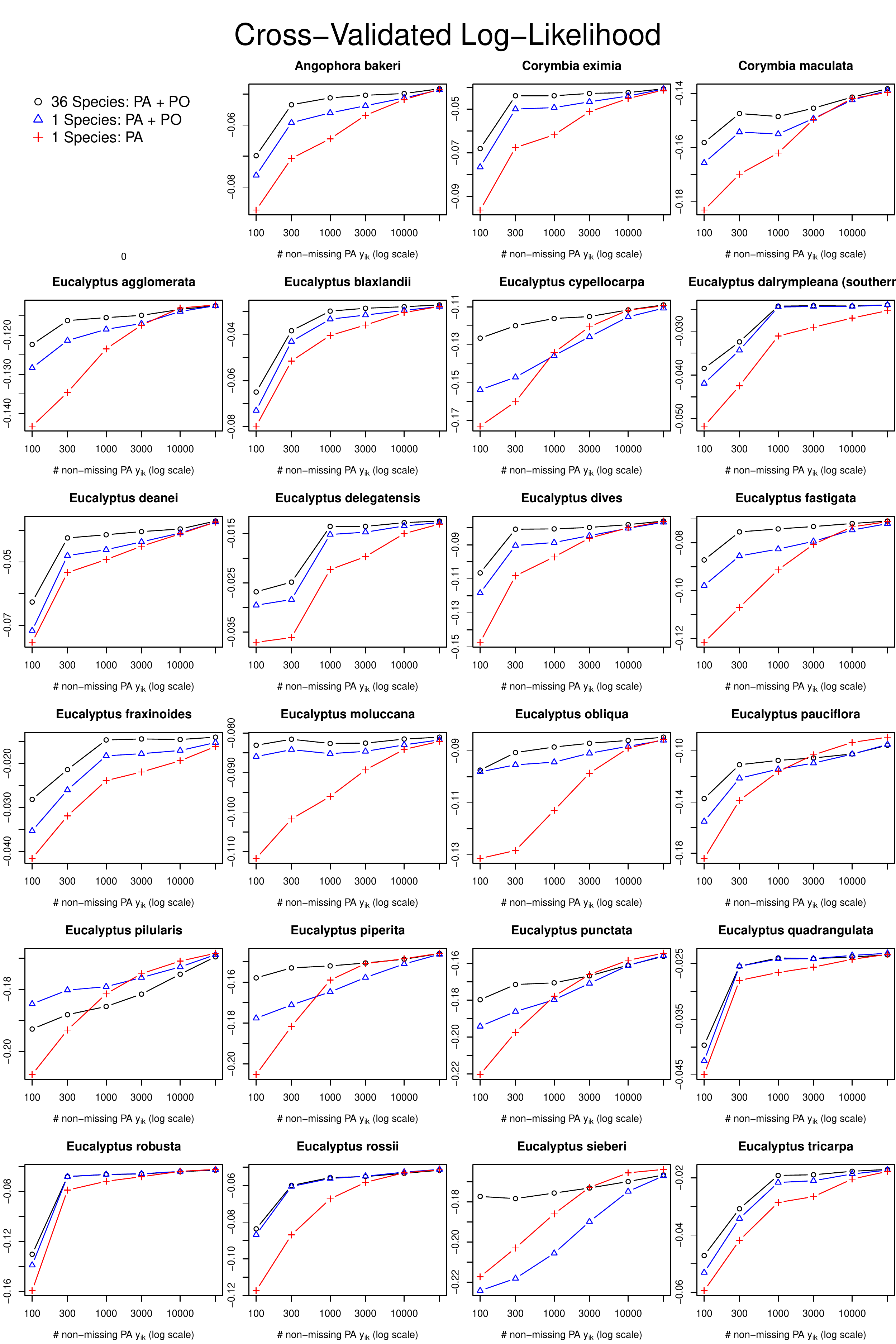}
  \caption{Cross-validation results for all species that were
    observed in at least 110 different presence-absence sites.
    Results vary for different species and different methods, but the
    method that pooled data across all 36 species had consistently
    superior performance.}
  \label{fig:allLoglik}
\end{figure}

\begin{figure}
  \centering
  \includegraphics[width=\textwidth]{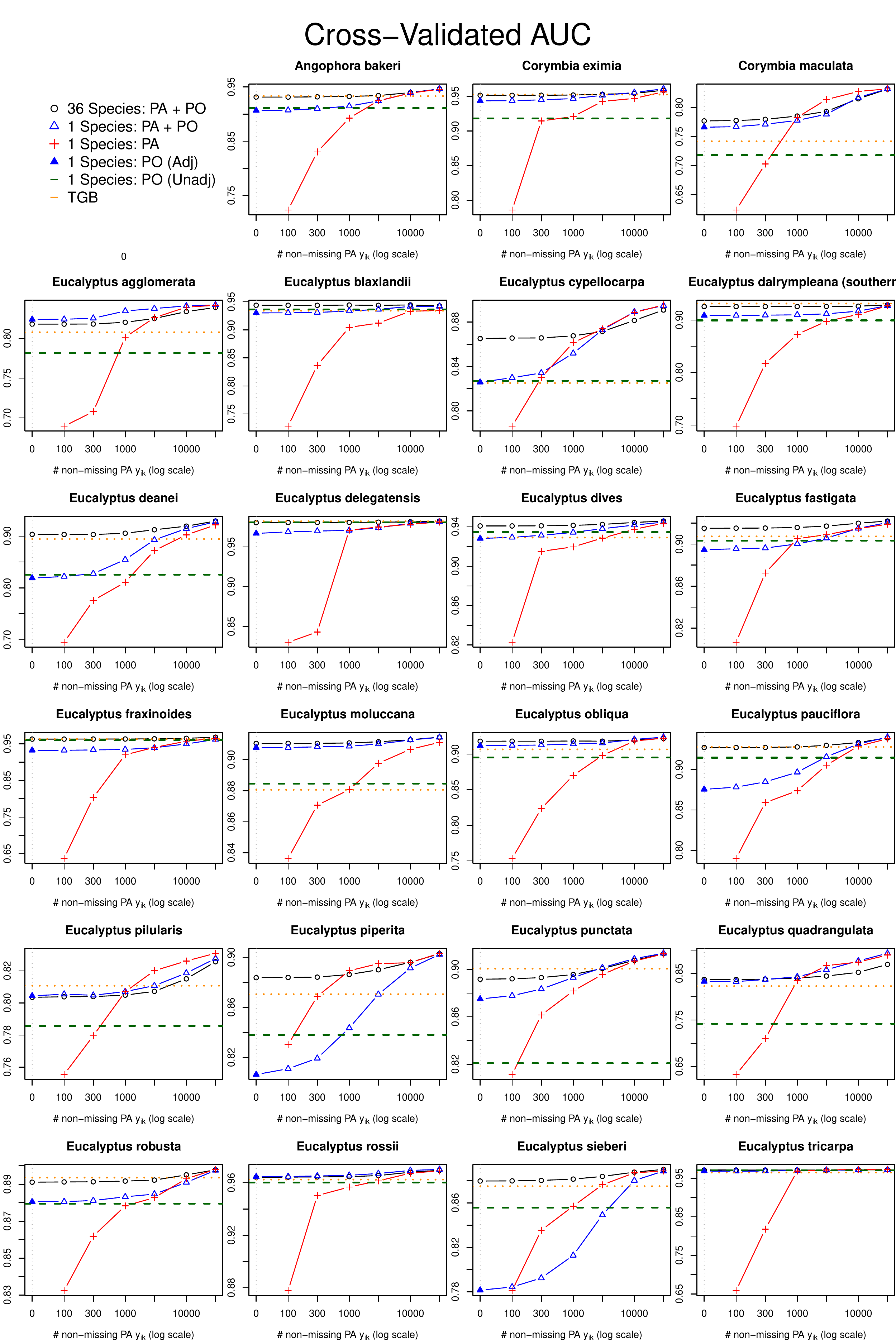}
  \caption{Cross-validation results for all species that were
    observed in at least 110 different presence-absence sites.
    Results vary for different species and different methods, but the
    method that pooled data across all 36 species had consistently
    superior performance.}
  \label{fig:allAUC}
\end{figure}



\section*{Supplementary material}

## Appendix C: Description of Data

Our target region is defined by "IBRA" regions (Interim Biogeographic Regionalisation of Australia) in the state of New South Wales (NSW) (Figure C1). Our selected IBRA regions cover tbe NSW coast and escarpments, tablelands and inland slopes of the Great Dividing Range. Species were selected based on several criteria: (a) likely to be detected if present; (b) likely to be correctly identified; (c) a large proportion of the extant species range in the selected IBRA regions; (d) survey data (presence-absence) well known to one of the authors (DK).

| Colour in map | Name |
|---|---|
| Chartreuse | Sydney Basin |
| Blue | South East Highlands (ACT & NSW) |
| Green | South East Corner |
| Dark blue | Australian Alps (ACT & NSW) |
| Yellow | North Coast (NSW) |
| Brown | New England Tablelands (NSW) |
| Orange | South East Queensland (NSW) |

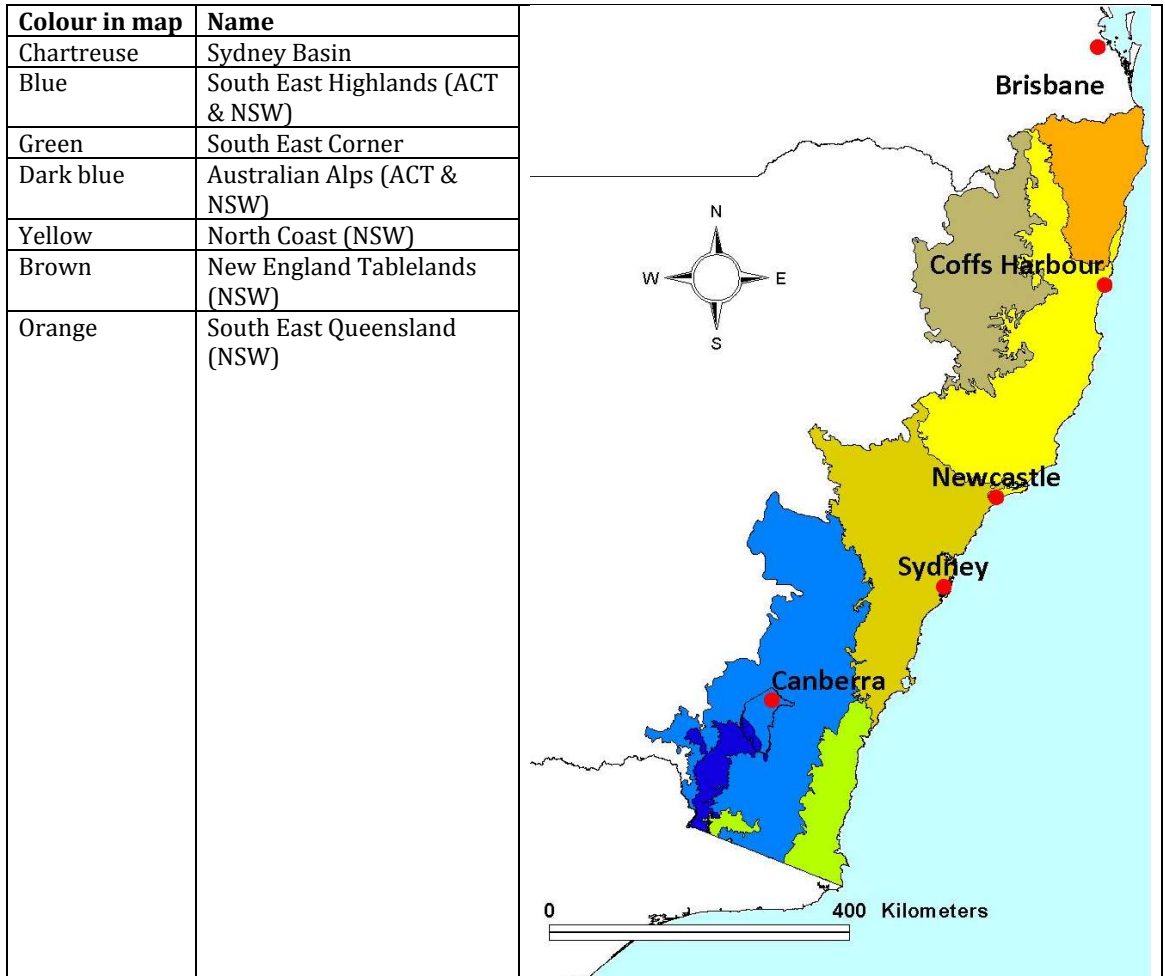

Figure C1: IBRA regions used in the study.

All species data were collated July to August 2013. Presence-absence (PA) data were downloaded from the Flora Survey Module of the Atlas of NSW Wildlife, Office of Environment and Heritage (OEH). All "full floristics" data for all selected IBRA regions were downloaded.  All quadrats were retained except for 10 recorded prior to 1970 and 122 whose locations had missing environmental data. This left 32612 quadrats, with 7 to 2003 presence records per species (Table C1)

Presence-only (PO) data were sourced from two repositories:

1. "NSW Atlas" - records from the Atlas of NSW Wildlife, Office of Environment and Heritage (OEH). The full record set included presence records from the PA data; PA presences were removed using the "LocationKey" unique identifier. Species names as per Table C1 were used; downloads included all subspecies. This Atlas can be accessed at  http://www.bionet.nsw.gov.au/ (last accessed July 2014)

2. "ALA" – records from the Atlas of Living Australia, excluding those with data provider = OEH. See Table C1 for details of any decisions about species nomenclature.  This atlas can be accessed at http://www.ala.org.au/ (last accessed July 2014).

These two data sources were combined, then PO data further "cleaned" as follows:

- no records prior to 1970 were retained, since older records tend to have poorer locational accuracy
- all recorded accuracies were initially retained.  This decision was based on the knowledge that accuracies are assigned in different ways depending on the data provider, sometimes based on record date rather than on true accuracy information. Visual checks of random subsets of the data across accuracy classes led us to conclude that records assigned a less accurate location class were no more likely to be in unlikely locations than those with supposedly accurate locations (basis for judgement: expert knowledge of species and their distribution: DK)
- for each species, records were reduced to one per unique location, where "unique" = rounded to closest metre. This decision reflects the fact that atlas records are submitted from multiple providers, and multiple records of the same species are possible. For instance, a collector might take several specimens from a tree and lodge them in different herbaria. Similarly, the same tree might be recorded by different people or at different times.
- all records for all species were visually examined in a GIS by DK and clearly erroneous records deleted
- records for *Eucalyptus dalrympleana* were divided into northern and southern "species" (at latitude -32.5 degrees) because they are known to have different environmental associations.

**Table C1: Details of species  (table continues over page)**

| Species name | Species code | # Presences in PA | # Presences in PO | Comments |
|---|---|---|---|---|
| *Angophora bakeri* | angobake | 439 | 508 | |
| *Corymbia eximia* | coryexim | 438 | 502 | |
| *Corymbia maculata* | corymacu | 1388 | 1806 | |
| *Eucalyptus agglomerata* | eucaaggl | 1025 | 1009 | |
| *Eucalyptus aggregata* | eucaaggr | 22 | 154 | |
| *Eucalyptus blaxlandii* | eucablax | 225 | 183 | |
| *Eucalyptus cinerea* | eucacine | 55 | 171 | |
| *Eucalyptus cypellocarpa* | eucacype | 1290 | 1536 | |
| *Eucalyptus dalrympleana - northern subsp* | eucadalh | 86 | 493 | North of -32.45 latitude |
| *Eucalyptus dalrympleana - southern subsp* | eucadalr | 172 | 1674 | South of -32.45 latitude |
| *Eucalyptus deanei* | eucadean | 304 | 182 | |
| *Eucalyptus delegatensis subsp. delegatensis* | eucadeld | 112 | 271 | |

| | | | | |
|---|---|---|---|---|
| *Eucalyptus dives* | eucadive | 905 | 1103 | |
| *Eucalyptus fastigata* | eucafast | 753 | 993 | |
| *Eucalyptus fraxinoides* | eucafrax | 125 | 339 | |
| *Eucalyptus gregsoniana* | eucagreg | 7 | 121 | |
| *Eucalyptus luehmanniana* | eucalueh | 41 | 358 | |
| *Eucalyptus moluccana* | eucamolu | 804 | 1007 | |
| *Eucalyptus niphophila* | eucaniph | 35 | 262 | Named E. pauciflora subsp niphophila in ALA |
| *Eucalyptus obliqua* | eucaobli | 953 | 847 | |
| *Eucalyptus obstans* | eucaobst | 31 | 166 | Named E. burgessiana in ALA |
| *Eucalyptus oreades* | eucaorea | 99 | 165 | |
| *Eucalyptus ovata* | eucaovat | 109 | 209 | |
| *Eucalyptus parramattensis* | eucaparr | 19 | 1438 | |
| *Eucalyptus parvula* | eucaparv | 7 | 116 | |
| *Eucalyptus pauciflora* | eucapauc | 1094 | 1489 | |
| *Eucalyptus pilularis* | eucapilu | 1777 | 2437 | |
| *Eucalyptus piperita* | eucapipe | 1762 | 1633 | |
| *Eucalyptus punctata* | eucapunc | 2103 | 1451 | |
| *Eucalyptus quadrangulata* | eucaquad | 141 | 320 | |
| *Eucalyptus robusta* | eucarobu | 538 | 900 | |
| *Eucalyptus rossii* | eucaross | 613 | 674 | |
| *Eucalyptus sieberi* | eucasieb | 2003 | 2483 | |
| *Eucalyptus squamosa* | eucasqua | 63 | 194 | |
| *Eucalyptus stenostoma* | eucasten | 18 | 90 | |
| *Eucalyptus tricarpa* | eucatric | 149 | 204 | |

**Environmental data** included climatic, topographic and substrate variables (Table C2) with 9 second (~ 250m x 250m) grid cells, unprojected.

Table C2. Candidate covariates for species distributions (table continues over page). Note the climate variables are long-term averaged data, using data supplied with ANUCLIM version 6.1 (ANU 2014) and estimated to 9 arc-second based on GeoScience Australia's 9 second Digital Elevation Model.

| Variable code | Variable | units | Longer explanation or comment |
|---|---|---|---|
| bc02 | Mean Diurnal Temperature Range | degrees C | Mean Diurnal Range (mean(period max-min)) - The mean of all the weekly diurnal temperature ranges. Each weekly diurnal range is the difference between that week's maximum and minimum temperature. |
| bc04 | Temperature Seasonality (C of V) | dimensionless | Temperature Seasonality (C of V) - The temperature Coefficient of Variation (C of V) is the standard deviation of the weekly mean temperatures expressed as a percentage of the mean of those temperatures (i.e. the annual mean). |
| bc05 | Maximum Temperature of Warmest Period | degrees C | Maximum Temperature of Warmest Period - The highest temperature of any weekly maximum temperature. |
| bc12 | Annual Precipitation | mm | Annual Precipitation - The sum of all the monthly precipitation estimates. |
| bc14 | Precipitation of Driest Period | mm | Precipitation of Driest Period - The precipitation of the driest week |
| bc21 | Highest Period Radiation | W/m2/day | Highest Period Radiation - The largest radiation estimate for all weeks. |
| bc32 | Mean Moisture Index of Highest Quarter MI | index | Mean Moisture Index of Highest Quarter MI - The quarter of the year having the highest moisture index value is determined (to the nearest week), and the average moisture index value is calculated. |
| bc33 | Mean Moisture Index of Lowest Quarter MI | index | Mean Moisture Index of Lowest Quarter MI - The quarter of the year having the lowest moisture index value is determined (to the nearest week), and the average moisture index value is calculated. |
| mvbf | Mutliresolution valleybottom flatness | index | MVBF classifies degrees of valley bottom flatness based on integrating estimates of 'flatness' and 'lowness' computed at a range of scales. MVBF is an expression of local relief in terms of valley confinement and floodplain extent with values typically ranging from 2.5 in narrow confined valleys to ≥ 8 in broad floodplains. Threshold values of 4-4.5 are often used to designate floodplains. See Gallant & Dowling 2003 |
| rjja | rain june july aug | mm | Precipitation in June, July, and August |

| Variable code | Variable | units | Longer explanation or comment |
|---|---|---|---|
| rseas | annual rainfall seasonality - warm (+ve) or cool (-ve) rainfall-dominated period | magnitudes | Annual rainfall seasonality is the factor variable in which warm-season-dominated rainfall is the ratio + warm-season/cool-season, and cool-season-dominated rainfall is the ratio -(minus sign) cool-season/warm-season; where the warm season rainfall is defined as the sum of rainfall over the six months Oct-Nov-Dec-Jan-Feb-Mar and the cool season rainfall is defined as the sum of rainfall over the six months Apr-May-Jun-Jul-Aug-Sep |
| rugg | Ruggedness | | Std deviation of elevation.in a 5 by 5 cell square centred on the cell of interest. Estimated on the 9 second DEM. |
| subs | Substrate class | | Based on the 1:1 miliion surface geology of Australia and Keith 2011, but – using Keith's knowledge of local conditions - collapsed to classes 0 (unclassified), 2 (siliceous white sandstones), 5 (low quartz primarily sedimentary), 6 (felsic intrusives), 7 (high quartz sedimentary), 8 (mafic volcanics & intrusives), 11 (floodplain, estuarine and lacustrine alluvium & sediments), 12 (residual alluvial/colluvial sand & gravel), 15 (felsic volcanics ) (see footnote 1) |
| twmd | Median topographic wetness | | Topographic wetness was estimated from a finer grain (~30m) digital elevation model (DEM) by John Gallant, CSIRO, then summarised to 9 second as median or maximum. |
| twmx | Maximum topographic wetness | | |

[1] Substrate classes.  Note that class 5 (previously "low quartz sedimentary") includes ex-categories 14 (ultramafic igneous & metamorphics), 19 (aeolian (red) sandplains) and 20 (limestone); class 12 includes 9 (residual alluvial sands); class 11 (previously "floodplain alluvium") includes 10 (estuarine sediments) and 13 (lacustrine sediments)

**References**
ANU (2014). Australian National University, ANUCLIM version 6.1.
http://fennerschool.anu.edu.au/research/products.

Gallant, J. C. and T. I. Dowling (2003). "A multiresolution index of valley bottom flatness for mapping depositiona areas." Water Resources Research **39**(12): 4-1 - 4-13.

Keith, D. A. 2011. Relationships between geodiversity and vegetation in south-eastern Australia. Proceedings of the Linnean Society of New South Wales 132:5-26.



\end{appendix}

\end{document}